\documentclass[journal,onecolumn,draftcls,12pt]{IEEEtran}

\usepackage{color}
\definecolor{gray}{rgb}{0.5,0.5,0.5}

\usepackage{amsmath,amsthm}
\usepackage{amssymb,amsfonts}
\usepackage{graphicx}
\usepackage{url}
\usepackage{subfigure}
\usepackage{bm}
\usepackage{subeqnarray} 
\usepackage{relsize}  
\usepackage[ruled,lined,linesnumbered]{algorithm2e}
\usepackage{paralist} 
\usepackage{hyperref}
\usepackage{caption}

\newtheorem{theorem}{Theorem}

\makeatletter
\def\IEEElabelanchoreqn#1{\bgroup
	\def\@currentlabel{\p@equation\theequation}\relax
	\def\@currentHref{\@IEEEtheHrefequation}\label{#1}\relax
	\Hy@raisedlink{\hyper@anchorstart{\@currentHref}}\relax
	\Hy@raisedlink{\hyper@anchorend}\egroup}
\makeatother
\newcommand{\subnumberinglabel}[1]{\IEEEyesnumber
	\IEEEyessubnumber*\IEEElabelanchoreqn{#1}}

\usepackage{cite}
\ifCLASSINFOpdf

\else

\fi

\begin{document}

\title{An Information Geometry Interpretation for Approximate Message Passing}

\author{Bingyan Liu, An-An Lu, \IEEEmembership{Member, IEEE}, Mingrui Fan, Jiyuan Yang, \IEEEmembership{Member, IEEE} and Xiqi Gao, \IEEEmembership{Fellow, IEEE}

\thanks{B. Liu was with the National Mobile Communications Research Laboratory (NCRL), Southeast University, Nanjing, 210096 China, and is now also with ASR Microelectronics Co., Ltd, Shanghai, China, e-mail: byliu@seu.edu.cn.} 

\thanks{A.-A. Lu, M. Fan, J. Yuan and X. Q. Gao are with the National Mobile Communications Research Laboratory (NCRL), Southeast University,
Nanjing, 210096 China, and also with Purple Mountain Laboratories, Nanjing 211111, China, e-mail: aalu@seu.edu.cn, xqgao@seu.edu.cn.} 
}


\maketitle

\begin{abstract}
	In this paper, we propose an information geometry (IG) framework to solve the 
	standard linear regression problem. 
	The proposed framework is an extension of the one for computing the mean of complex multivariate Gaussian distribution. By applying the proposed framework, 	
	the information geometry approach (IGA) and the approximate information geometry approach (AIGA) for basis pursuit de-noising (BPDN) in standard linear regression are derived. The framework can also be applied to other standard linear regression problems. 
	With the transformations of natural and expectation parameters of Gaussian distributions, we then show the relationship between the IGA and the message passing (MP) algorithm. Finally, we prove that the AIGA is equivalent to the approximate message passing (AMP) algorithm.  
	These intrinsic results offer a new perspective for the AMP algorithm, and clues for understanding and improving stochastic reasoning methods.
\end{abstract} 

%
\IEEEpeerreviewmaketitle

\newpage
\section{Introduction}
Information geometry \cite{amari2000methods, amari2016information, nielsen2020elementary} arises from the study of invariant geometric structures of families of probability distributions. It explores information sciences through modern geometric methods, and
views the space of probability distributions as differentiable manifolds.
By defining the Riemannian metric on them, it applies differential geometric tools to characterize and study geometric structures such as distance and curvature in the Riemannian space of probability distributions. Information geometry has a beautiful mathematical structure and provides insight into statistics and its applications \cite{ay2017information}. It is still a very active field, and its area of applications includes machine learning, signal processing, optimization theory \cite{amari2010information, raskutti2015information}, optimal transport theory \cite{li2022transport}, integrated information theory \cite{oizumi2016unified},  neuroscience, geometric mechanics \cite{leok2017connecting} and quantum physics \cite{jarzyna2020geometric}.

Information geometry has the potential to provide a unified framework for statistical inference methods. 
In \cite{ikeda2004stochastic}, an information geometry interpretation of the belief propagation (BP) algorithm and the concave-convex procedure (CCCP) algorithm \cite{yuille2002cccp}, proposed based on minimizing the Bethe free energy \cite{yedidia2001bethe}, for solving the marginal distribution of the \textit{a posteriori} distribution is given. 
In \cite{Ikeda2004}, the decoding algorithms for Turbo codes and low-sensitivity parity check (LDPC) codes are derived from the viewpoint of information geometry, and the equilibrium and error of the algorithms are analyzed. Recently, information geometry has been used in wireless communications to derive iterative information geometry algorithms (IGAs)  
for massive MIMO channel estimation \cite{Yang2022, yang2023channel} and ultra-massive MIMO signal detection \cite{yang2024signal}. Moreover, an information geometry framework and a new interference cancellation information geometry algorithm (IC-IGA) for computing the mean of a complex multivariate Gaussian distribution is proposed in \cite{lu2024ICIGA}.

Meanwhile, classical statistical inference methods such as belief propagation (BP) \cite{Pearl1988} and message passing (MP) are developed from graphical models.
To cover more scenarios and solve more general problems, approximate message passing (AMP) \cite{Donoho2010}, expectation propagation (EP) \cite{minka2013expectation}, generalized approximate message passing (GAMP) \cite{Rangan2011}, and Bethe free energy minimization \cite{Zhang2021} are evolved from them. These algorithms have a wide range of applications in various areas. The equivalence of GAMP, EP and Bethe free energy minimization has been shown in \cite{Thomas2023, Zhang2021}.
To enhance the convergence behavior of the AMP, various variants of AMP such as orthogonal AMP (OAMP)\cite{Ma2017, Liul2023}, convolutional AMP (CAMP)\cite{Takeuchi2021}, memory AMP (MAMP) \cite{Liu2021} and vector AMP (VAMP)\cite{Rangan2017Vector} are proposed in the literature.

As we have shown, information geometry assists in a deeper understanding of statistical inference problems and might be able to improve the algorithms due to its intuitive geometric understanding. It seems natural to provide an information geometry framework for the AMP algorithm as well. However, such work has not been provided and also can not be easily extended from those in the literature. To fill the gap, we provide an information geometry interpretation of the AMP algorithm in this paper. To achieve this goal, we first propose a general framework for the standard linear regression problem with non-Gaussian prior. We then derive the IGA  and the AIGA based on the framework. Finally, we show the equivalence between the AIGA and the AMP algorithm. 

The rest of this paper is organized as follows. 
The problem formulation and information geometry framework are given in Section \ref{sec:ProIGf}.
The information geometry approach for BPDN is proposed in Section \ref{sec:IGA}. 
The approximate information geometry approach for BPDN is proposed in Section \ref{sec:AIGA}. 
The relationship between the IGA and the MP algorithm, as well as the proof of equivalence of the AIGA and the AMP algorithm are presented in Section \ref{sec:AMPeqIGA}. 
Simulation results are provided in Section \ref{sec:Sim}. 
The conclusion is drawn in Section \ref{sec:Con}.

{\it Notations}: Throughout this paper, uppercase and lowercase boldface letters are used to represent matrices and vectors. The superscripts $(\cdot)^*$ and $(\cdot)^T$  denote the conjugate and transpose operations. Operators ${\mathbb E}\{\cdot\}$ and ${\mathbb V}\{\cdot\}$ denote the mathematical expectation and variance, respectively. 
The operator $\text{det}(\cdot)$ represents the matrix determinant. The $\ell_1$ and $\ell_2$ norms are denoted by $\|\cdot\|_1$ and $\|\cdot\|_2$, respectively. The operator $\odot$ denotes the Hadamard product. The $N \times N$ identity matrix is denoted by $\mathbf I_N$, where the subscript $N$ is sometimes omitted for convenience.
We use $h_n$ or $[\mathbf h]_n$, $a_{mn}$ or $[\mathbf{A}]_{mn}$, $[\mathbf{A}]_{:,n}$ and $[\mathbf{A}]_{m,:}$ to denote the $n$-th element of the vector $\mathbf h$, the $(m,n)$-th element of the matrix $\mathbf{A}$, the $n$-th column and the $m$-th row of matrix $\mathbf{A}$.

\section{Problem Formulation and Information Geometry Framework}
\label{sec:ProIGf}

\subsection{Problem Formulation}
\label{sec:ProFor}
Consider a standard linear regression problem on sparse signal recovery. The receive model for the problem is given as
\begin{IEEEeqnarray}{Cl}\label{reModel}
	\mathbf y= \mathbf A\mathbf h+ \mathbf z
\end{IEEEeqnarray}
where 
$\mathbf y \in \mathbb R^{M \times 1}$ is the observation, 
$\mathbf A \in \mathbb R^{M \times N}(M<N)$ is the measurement matrix, 
$\mathbf h \in \mathbb R^{N \times 1}$ is the sparse signal to be recovered, 
$\mathbf z \in \mathbb R^{M \times 1}$ is the additive white Gaussian noise with zero mean and $\sigma_z^2$ covariance. 
When the \textit{a priori} statistics of $\mathbf h$ and noise power are both unknown, this sparse signal recovery problem is a classical compressed sensing problem, also known as the basis pursuit de-noising (BPDN) or the least absolute shrinkage and selection operator (LASSO) inference problem. 
The problem is then modeled as
\begin{IEEEeqnarray}{Cl}\label{eq:minNorm}
	\min_{\mathbf h} \left\{ \kappa\|\mathbf h\|_1 + \frac{1}{2}\|\mathbf y-\mathbf A\mathbf h\|_2^2 \right\}
\end{IEEEeqnarray}
where $\kappa > 0$ is a regularization parameter that balances the sparsity and the estimation error.
Further, this problem can be transformed into
\begin{IEEEeqnarray}{Cl}\label{eq:expmax}
	\max_{\mathbf h} \left\{\exp\left\{  - \beta \left(\kappa \|\mathbf h\|_1 
	+ \frac{1}{2}\|\mathbf y-\mathbf A\mathbf h\|_2^2\right) \right\} \right\}
\end{IEEEeqnarray}
where $\beta>0$. According to \cite{Donoho2010}, when $\beta \rightarrow \infty$, problems \eqref{eq:minNorm} and \eqref{eq:expmax} have the same solution.

The exponential form constructed in \eqref{eq:expmax} is similar to the PDF of an exponential family distribution. 
Let the postulated \textit{a priori} PDF of $\mathbf h$ be defined by
\begin{IEEEeqnarray}{rl}
	p(\mathbf h) 
	&\propto \exp\left\{-\beta\kappa\|\mathbf h\|_1\right\} 
\end{IEEEeqnarray}
and the conditional PDF of $\mathbf y$ given $\mathbf h$ be defined by 
\begin{IEEEeqnarray}{rl}
	p(\mathbf y|\mathbf h) 
	&\propto \exp\left\{- \frac{\beta}{2}\|\mathbf y-\mathbf A\mathbf h\|_2^2\right\}.
\end{IEEEeqnarray}
By using Bayes' theorem, we can construct the postulated \textit{a posteriori} PDF of $\mathbf h$ as
\begin{IEEEeqnarray}{rl}
	p(\mathbf h|\mathbf y) 
	&\propto p(\mathbf h) p(\mathbf y|\mathbf h) \notag\\
	&\propto \exp\left\{-\beta\kappa\|\mathbf h\|_1 - \frac{\beta}{2}(\mathbf y-\mathbf A\mathbf h)^T(\mathbf y-\mathbf A\mathbf h)\right\}.
	\IEEEeqnarraynumspace
\end{IEEEeqnarray}
In the above formula, the \textit{a priori} distribution is a Laplace distribution. For other more general \textit{a priori} distributions, 
$-\beta\kappa\|\mathbf h\|_1$ becomes $d(\mathbf h)$, where $d(\mathbf h) = \sum_n d_n(h_n)$  and its specific form depends on specific problems. We rewrite the \textit{a posteriori} PDF of $\mathbf h$ as
\begin{IEEEeqnarray}{Cl}
	p( \mathbf h| \mathbf y)  
	&\propto \exp\left\{ d(\mathbf h)
	- \frac{\alpha}{2}\mathbf h^T\mathbf A^T\mathbf A\mathbf h 
	+ {\alpha}\mathbf h^T\mathbf A^T\mathbf y \right\}  
\end{IEEEeqnarray}
for the general problem. The object is to find the mean of the distribution $p(\mathbf h|\mathbf y)$.
In this paper, we solve this problem with information geometry theory and illustrate the relation between the proposed information geometry algorithms with the MP algorithms.

\subsection{Information Geometry Framework}
\label{sec:IGAproce}
In this subsection, we present the information geometry framework for standard linear regression.
The object is to find an approximated mean of the \textit{a posteriori} distribution with an iterative method.
The framework is an extension of the work in \cite{lu2024ICIGA}, which is inspired by the explanation of the BP algorithm with information geometry in \cite{amari2016information}.

The \textit{a posteriori} PDF can be rewritten as
\begin{IEEEeqnarray}{Cl}\label{eq:expost}
	p_d( \mathbf h;  \hat{\bm\theta}_{or}, \hat{\bm\Theta}_{or}) 
	\propto 
	\exp\left\{ d(\mathbf h)
	+ \mathbf h^T \hat{\bm\Theta}_{or}\mathbf h
	+ \mathbf h^T \hat{\bm\theta}_{or}
	\right\}  
\end{IEEEeqnarray}
where $\hat{\bm\theta}_{or} = {\alpha}\mathbf{A}^T\mathbf y$ and $\hat{\bm\Theta}_{or} = - \frac{\alpha}{2} \mathbf A^T\mathbf A$. 

The information geometry framework in \cite{lu2024ICIGA} is proposed for computing the mean of a complex multivariate Gaussian distribution. However, it can not be applied directly to the distribution in \eqref{eq:expost}, since the \textit{a priori} distribution is not Gaussian. Thus, we propose a general framework for this case as summarized below.

\begin{compactenum}
	\item[\quad (1)] The parameters $\hat{\bm\theta}_{or}$ and $\hat{\bm\Theta}_{or}$ are split to construct $Q$ auxiliary manifolds of PDFs.
	The \textit{a priori} term $d(\mathbf h)$ is used to construct an extra auxiliary manifold of PDFs, and one target manifold of PDFs is also constructed;
	\item[\quad (2)]Initialize the parameters for auxiliary points and the target point in auxiliary manifolds and the target manifold with the $e$-condition;  
	\item[\quad (3)]Compute the beliefs of the auxiliary points based on the $m$-projection;
	\item[\quad (4)]Update the natural parameters of the auxiliary and target points with the beliefs and a condition that the extra auxiliary point and the target point have the same mean and variance;
	\item[\quad (5)]Repeat (3) and (4) until the algorithm converges or fixed iterations.
	Output the mean and variance of the target point.
\end{compactenum}
We give the detailed framework in the following.

\subsubsection{Auxiliary manifolds and the $e$-condition}

The first step of the information geometry framework is the split of $\hat{\bm\theta}_{or}$ and $\hat{\bm\Theta}_{or}$.
By splitting $\hat{\bm\theta}_{or}$ and $\hat{\bm\Theta}_{or}$ into $Q$ items, we have
\begin{IEEEeqnarray}{Cl}
	\subnumberinglabel{eq:IGsplit}
	\hat{\bm\theta}_{or} = {\alpha} \mathbf{A}^T\mathbf y 
	=  \sum_{q=1}^Q \mathbf{b}_q \\
	\hat{\bm\Theta}_{or} = - \frac{\alpha}{2} \mathbf{A}^T\mathbf{A}
	= - \sum_{q=1}^Q \mathbf{C}_q 
\end{IEEEeqnarray}
where the setting of $\mathbf{b}_q$, $\mathbf{C}_q$ and $Q$ depends on specific algorithms.

Based on the splits, we define $Q$ auxiliary PDFs.
The $q$-th auxiliary PDF is defined by
\begin{IEEEeqnarray}{Cl}\label{eq:auxpost}
	&\quad p( \mathbf h;  {\bm\theta}_q, {\bm\Theta}_q) \propto 
	\exp\left\{- \mathbf h^T({\bm\Lambda}_q + \mathbf C_q)\mathbf h
	+ \mathbf h^T({\bm\lambda}_q + \mathbf b_q)
	\right\}  
\end{IEEEeqnarray}
where $\bm\Lambda_q$ is diagonal and 
\begin{IEEEeqnarray}{Cl}
	\subnumberinglabel{eq:tThetaq}
	{\bm\theta}_q &= {\bm\lambda}_q+\mathbf b_q
	\label{eq:thetaq}\\
	{\bm\Theta}_q &= - ({\bm\Lambda}_q + \mathbf C_q)
	\label{eq:Thetaq}
\end{IEEEeqnarray}
 are natural parameters of $p( \mathbf h;  {\bm\theta}_q, {\bm\Theta}_q)$. 
We define the $q$-th auxiliary manifold as 
\begin{IEEEeqnarray}{Cl}
	{\mathcal M}_q
	= \left\{ p( \mathbf h;  {\bm\theta}_q, {\bm\Theta}_q)  \right\}.
\end{IEEEeqnarray}
As illustrated in \cite{lu2024ICIGA}, ${\mathcal M}_q$ is a submanifold of the dually flat manifold $\cal{M}$ of Gaussian distributions, and the PDF $p( \mathbf h;  {\bm\theta}_q, {\bm\Theta}_q)$
can be viewed as a point with the affine coordinate $({\bm\theta}_q, {\bm\Theta}_q)$ in ${\mathcal M}_q$.  The affine coordinate is called the $e$-coordinate.
Because the natural parameters are linearly constrained in $e$-coordinate, ${\mathcal M}_q$ is an $e$-flat submanifold.

To approximate the term $d(\mathbf h)$, an extra auxiliary PDF  $p(\mathbf h;  \hat{\bm\theta}_0, \hat{\bm\Theta}_0)$ is defined as 
\begin{IEEEeqnarray}{Cl}
	\label{eq:tarPDF}
	p_d( \mathbf h;  \hat{\bm\theta}_0, \hat{\bm\Theta}_0)
	\propto \exp\left\{ d(\mathbf h)
	- \mathbf h^T\hat{\bm\Lambda}_0 \mathbf h 
	+  \mathbf h^T\hat{\bm\lambda}_0 \right\}  
	\IEEEeqnarraynumspace
\end{IEEEeqnarray}
where $\hat{\bm\theta}_0, \hat{\bm\Theta}_0$ are defined as 
\begin{IEEEeqnarray}{ClCl}
	\hat{\bm\theta}_0 &=  \hat{\bm\lambda}_0 , \quad
	& \hat{\bm\Theta}_0 &= -  \hat{\bm\Lambda}_0.
\end{IEEEeqnarray}
The corresponding auxiliary manifold is defined as
\begin{IEEEeqnarray}{Cl}
	\hat{\mathcal M}_0
	= \left\{ p_d( \mathbf h;  \hat{\bm\theta}_0, \hat{\bm\Theta}_0)  \right\} 
\end{IEEEeqnarray}
which does not belong to the dually flat manifold $\cal{M}$.

The object of the problem is the marginal distribution of $\mathbf h$, which is approximated with the target point in the information geometry framework.
The target PDF $p( \mathbf h; \bm\theta_0 , \bm\Theta_0)$ is defined as
\begin{IEEEeqnarray}{Cl}
	\label{eq:tarpoint}
	p(\mathbf h; \bm\theta_0 , \bm\Theta_0) 
	\propto \exp\left\{
	- \mathbf h^T{\bm\Lambda}_0 \mathbf h 
	+  \mathbf h^T{\bm\lambda}_0 \right\}  
\end{IEEEeqnarray}
where ${\bm\Lambda}_0$ is diagonal, which means $p( \mathbf h; \bm\theta_0 , \bm\Theta_0)$ contains no interaction item, thus $h_i$ and $h_j$ are independent for $i\neq j$. The natural parameters of $p( \mathbf h; \bm\theta_0 , \bm\Theta_0) \in \mathcal M_0$ are
\begin{IEEEeqnarray}{ClCl}
	\bm\theta_0 &=  {\bm\lambda}_0 , \quad
	& \bm\Theta_0 &= -  {\bm\Lambda}_0.
\end{IEEEeqnarray}
The items corresponding to ${\bm\lambda}_0$ and ${\bm\Lambda}_0$ are the approximations of items corresponding to $\sum_{q} \mathbf b_q$ and $\sum_{q} \mathbf C_q$ adding the information in $d(\mathbf h)$.
Define the set of all possible target points $p( \mathbf h; \bm\theta_0 , \bm\Theta_0)$ as the target manifold
\begin{IEEEeqnarray}{Cl}
	\mathcal M_0
	= \left\{p( \mathbf h; \bm\theta_0 , \bm\Theta_0)\right\} 
\end{IEEEeqnarray}
which is also an $e$-flat manifold.

\begin{figure}
	\centering
	\includegraphics[width=0.6\linewidth]{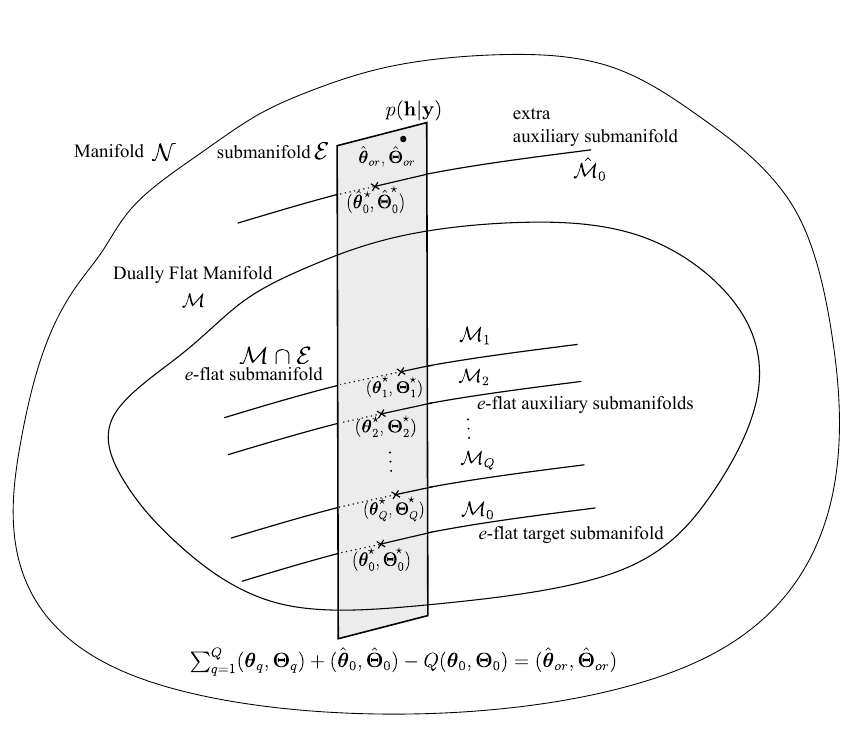}
	\caption{$e$-condition}
	\label{fig:econd}
\end{figure}

After defining the auxiliary manifolds and the target manifold, we introduce the first condition of the information geometry framework, which is given by
\begin{IEEEeqnarray}{Cl}
	\label{eq:econdnewth}
	\sum_{q=1}^Q({\bm\theta}_q,{\bm\Theta}_q) + (\hat{\bm\theta}_0,\hat{\bm\Theta}_0) - Q({\bm\theta}_0,{\bm\Theta}_0) = (\hat{\bm\theta}_{or},\hat{\bm\Theta}_{or})
	\IEEEeqnarraynumspace
\end{IEEEeqnarray}
or equivalently 
\begin{IEEEeqnarray}{Cl}
	\label{eq:econdnewla}
	\sum_{q=1}^Q({\bm\lambda}_q,{\bm\Lambda}_q) + (\hat{\bm\lambda}_0,\hat{\bm\Lambda}_0) - Q({\bm\lambda}_0,{\bm\Lambda}_0) = 0
\end{IEEEeqnarray}
which means $({\bm\lambda}_q,{\bm\Lambda}_q)$, $(\hat{\bm\lambda}_0,\hat{\bm\Lambda}_0)$ and $({\bm\lambda}_0,{\bm\Lambda}_0)$ are linearly dependent.
It is an extension of the $e$-condition defined in \cite{lu2024ICIGA}.
We still call it the $e$-condition and plot it in Fig.~\ref{fig:econd} for illustration purpose.

\subsubsection{The $m$-projection and the $m$-condition}

Let $\bm\mu_q$ and $\bm\Sigma_q$ be the mean and covariance of the Gaussian PDF $p( \mathbf h;  {\bm\theta}_q, {\bm\Theta}_q)$, respectively, which are also called expectation parameters.  Let $\hat{\bm\mu}_0$ and $\hat{\bm\Sigma}_0$ be the mean and covariance of $p_d( \mathbf h;  \hat{\bm\theta}_0, \hat{\bm\Theta}_0)$, respectively.
By transforming the natural parameters of auxiliary points into expectation parameters, we have
\begin{IEEEeqnarray}{Cl}
	\bm\mu_q 
	&= - \frac{1}{2} \bm\Theta_q^{-1}
	\bm\theta_q =  \frac{1}{2} ({\bm\Lambda}_q + \mathbf C_q)^{-1} ({\bm\lambda}_q+\mathbf{b}_q ) 
	\IEEEyesnumber\IEEEyessubnumber*
	\label{eq:muqhat}\\
	\bm\Sigma_q
	&= - \frac{1}{2} \bm\Theta_q^{-1} =   \frac{1}{2} ({\bm\Lambda}_q + \mathbf C_q)^{-1}.
\end{IEEEeqnarray}
The projection onto an $e$-flat manifold is called the $m$-projection.
Then, we $m$-project auxiliary points into the target manifold. 
The $m$-projection of $p( \mathbf h;  {\bm\theta}_q, {\bm\Theta}_q)$ to $\mathcal M_0$ doesn't change the expectations of $\mathbf h$ \cite[Theorem 11.6]{amari2016information}.
From $p( \mathbf h; \bm\theta_0 , \bm\Theta_0) \in \mathcal M_0$, the mean $\bm\mu_q^0$ and covariance $\bm\Sigma_q^0$ of the $m$-projection satisfy 
\begin{IEEEeqnarray}{ClCl}
	\label{eq:muMmproj0}
	\bm\mu_q^0
	&= \bm\mu_q , \quad
	& \bm\Sigma_q^0
	&= \mathbf I \odot \bm\Sigma_q.
\end{IEEEeqnarray}

The natural parameters of the $m$-projection points are transformed from the expectation parameters as
\begin{IEEEeqnarray}{Cl}
	\subnumberinglabel{eq:projtTheta}
	\bm\theta_q^0
	&= (\bm\Sigma_q^0)^{-1} \bm\mu_q^0
	= - 2 \cdot \bm\Theta_q^0\bm\mu_q^0
	\notag\\
	&= - \bm\Theta_q^0
	({\bm\Lambda}_q + \mathbf C_q)^{-1} ({\bm\lambda}_q+\mathbf{b}_q )\\
	\bm\Theta_q^0
	&= -\frac{1}{2} (\bm\Sigma_q^0)^{-1}
	= (\mathbf I\odot \bm\Theta_q^{-1})^{-1}
	\notag\\
	&= - (\mathbf I\odot 
	({\bm\Lambda}_q + \mathbf C_q)^{-1})^{-1}.
\end{IEEEeqnarray}
Here, the $m$-projection still involves matrix inversion, \textit{i.e.}, $({\bm\Lambda}_q + \mathbf C_q)^{-1}$. 
Since ${\bm\Lambda}_q$ is diagonal, the matrix inversion can be implemented with low complexity by a proper setting of $\mathbf C_q$.

Define $\bm\theta_q^0
=  {\bm\lambda}_q^0$, 
$\bm\Theta_q^0
= - {\bm\Lambda}_q^0$.
The beliefs are calculated as
\begin{IEEEeqnarray}{Cl}
	\bm\xi_q 
	={\bm\lambda}_q^0 
	- {\bm\lambda}_q, \quad
	\bm\Xi_q 
	= {\bm\Lambda}_q^0
	- {\bm\Lambda}_q
\end{IEEEeqnarray}
where $\bm\Xi_q$ is also diagonal. 
Thus, the natural parameters of $m$-projection $p( \mathbf h; \bm\theta_q^0 , \bm\Theta_q^0)$ can also be represented as.
\begin{IEEEeqnarray}{Cl}
	\subnumberinglabel{eq:tThetaq0}
	\bm\theta_q^0
	&=  {\bm\lambda}_q + \bm\xi_q  \\
	\bm\Theta_q^0
	&= -\left({\bm\Lambda}_q + \bm\Xi_q  \right).
\end{IEEEeqnarray}
By comparing \eqref{eq:tThetaq0} with \eqref{eq:tThetaq}, it can be seen that $\bm\xi_q$,  $\bm\Xi_q$ are approximations of $\mathbf b_q$,  $\mathbf C_q$ in $\bm\theta_q$ and $\bm\Theta_q$.

The items corresponding to ${\bm\lambda}_q$, ${\bm\Lambda}_q$ are the approximations of the interaction items except the $q$-th interaction item adding the information in $d(\mathbf h)$. 
Then the update for parameters of the target point and auxiliary points are constructed as
\begin{IEEEeqnarray}{ClCl}
	\subnumberinglabel{eq:newlLambda0q}
	\hat{\bm\lambda}_0^{t+1}
	&= \sum_{q=1}^Q \bm\xi_q^t  , \quad
	& \hat{\bm\Lambda}_0^{t+1} 
	&= \sum_{q=1}^Q \bm\Xi_q^t  \\
	\bm\lambda_0^{t+1}  
	&=  \upsilon(\hat{\bm\lambda}_0^{t+1},
	\hat{\bm\Lambda}_0^{t+1}) 
	, \quad
	& \bm\Lambda_0^{t+1}  &=  \Upsilon(\hat{\bm\lambda}_0^{t+1},
	\hat{\bm\Lambda}_0^{t+1})\\
	\bm\lambda_q^{t+1}
	&= \bm\lambda_0^{t+1}
	- \bm\xi_q^t , \quad
	& \bm\Lambda_q^{t+1}
	&= \bm\Lambda_0^{t+1}
	- \bm\Xi_q^t 
	\IEEEeqnarraynumspace
\end{IEEEeqnarray}
where the superscript $t$ denotes the number of iterations, 
functions $\upsilon(\hat{\bm\lambda},
\hat{\bm\Lambda})$ and $\Upsilon(\hat{\bm\lambda},
\hat{\bm\Lambda})$ are designed with the criterion of making $(\bm\mu_0,\bm\Sigma_0)
	 =(\hat{\bm\mu}_0,\hat{\bm\Sigma}_0)$. 
The choice of $\upsilon(\hat{\bm\lambda},
\hat{\bm\Lambda})$ and $\Upsilon(\hat{\bm\lambda},
\hat{\bm\Lambda})$ depends on specific \textit{a priori} PDF $p(\mathbf h)$.
When the \textit{a priori} PDF is a Gaussian PDF, we have $\bm\lambda_0^{t+1}=\hat{\bm\lambda}_0^{t+1}$ and ${\bm\Lambda}_0^{t+1}=\hat{\bm\Lambda}_0^{t+1}$, and the update reduces to that in \cite{lu2024ICIGA}.
From \eqref{eq:newlLambda0q} we can easily verify 
that the $e$-condition holds.


%
%
%
%
%
%

By iteratively calculating $m$-projection, beliefs and parameter update, the parameters $\bm\theta_0^{\star}$, $\bm\Theta_0^{\star}$ of target point are finally obtained when the algorithm converges, thereby obtaining the target PDF $p( \mathbf h; \bm\theta_0^{\star}, \bm\Theta_0^{\star})$. 
From \eqref{eq:tThetaq0} and \eqref{eq:newlLambda0q}, when the algorithm converges, all the parameters no longer change, then we have 
\begin{IEEEeqnarray}{Cl}
	\subnumberinglabel{eq:tThetaq0star}
	\big({\bm\theta}_q^0\big)^{\star}
	&=  {\bm\lambda}_0^{\star}
	= {\bm\theta}_0^{\star}   \\
	\big({\bm\Theta}_q^0\big)^{\star}
	&= -  {\bm\Lambda}_0^{\star}  
	= {\bm\Theta}_0^{\star}.
\end{IEEEeqnarray}
%
Thus, the following condition is satisfied
\begin{IEEEeqnarray}{Cl}
	\label{eq:mcon}
	(\bm\mu_0^{\star},\bm\Sigma_0^{\star})
	 =(\hat{\bm\mu}_0^{\star},\hat{\bm\Sigma}_0^{\star}) = (\bm\mu_q^{\star} ,\mathbf I \odot \bm\Sigma_q^{\star} ). 
\end{IEEEeqnarray}
It is called the $m$-condition and is plotted in Fig.~\ref{fig:mcond}.
Finally, the output of this algorithm is the mean and covariance of target point
\begin{IEEEeqnarray}{Cl}
	\bm\mu_0 
	&= - \frac{1}{2} \bm\Theta_0^{-1}
	\bm\theta_0
	=  \frac{1}{2}{\bm\Lambda}_0^{-1}{\bm\lambda}_0
	\IEEEyesnumber\IEEEyessubnumber*\\
	\bm\Sigma_0
	&= - \frac{1}{2} \bm\Theta_0^{-1}
	= \frac{1}{2} {\bm\Lambda}_0^{-1}
\end{IEEEeqnarray}
which is regarded as the approximated mean and diagonal of the covariance matrix of the original \textit{a posteriori} PDF.


 \begin{figure}
	\centering
	\includegraphics[width=0.6\linewidth]{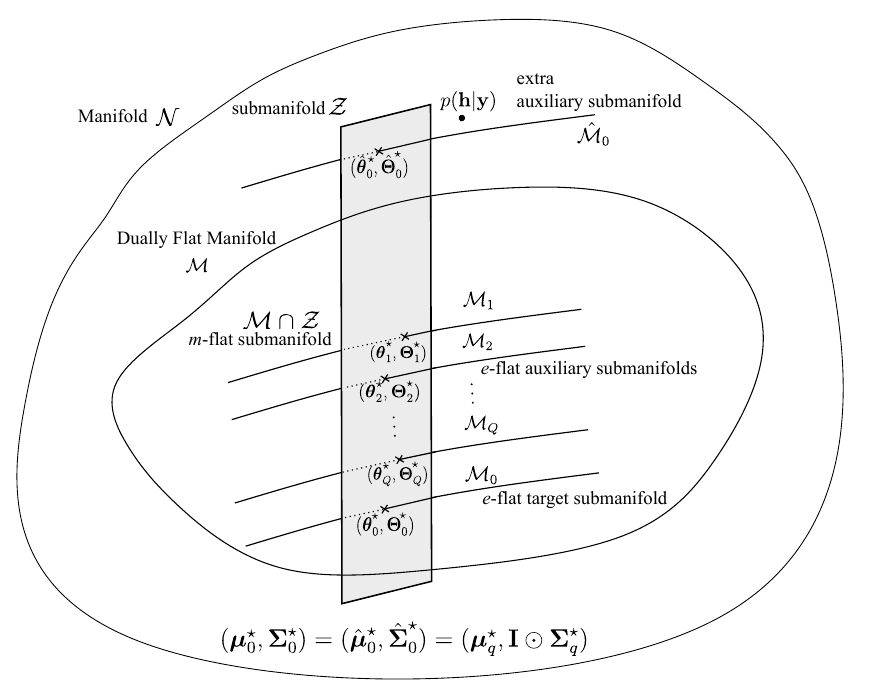}
	\caption{$m$-condition}
	\label{fig:mcond}
\end{figure}

\section{Information Geometry Approach for BPDN}
\label{sec:IGA}
In this section, we provide the derivation of IGA for BPDN in standard linear regression. 
The IGA for basis pursuit and Bayesian inference can be derived similarly, and thus are omitted for brevity.

\subsection{Auxiliary and Target PDFs}

Recall the received signal model \eqref{reModel}. 
The postulated \textit{a posteriori} PDF of $\mathbf h$  is expressed as
\begin{IEEEeqnarray}{Cl}
	\label{eq:postupost}
	&\quad p( \mathbf h|\mathbf y) = \exp\left\{-\beta\kappa\|\mathbf h\|_1 
	- \frac{\beta}{2}\mathbf h^T\mathbf A^T\mathbf A\mathbf h 
	+ \beta\mathbf h^T\mathbf A^T\mathbf y 
	+ \psi_{or}\right\}
	\IEEEeqnarraynumspace
\end{IEEEeqnarray}	
where $\psi_{or}$ is the normalization factor. 
We split $\frac{\beta}{2}\mathbf A^T\mathbf A$ and $\beta\mathbf A^T\mathbf y$ based on each element of the received vector $\mathbf y$, \textit{i.e.},
\begin{IEEEeqnarray}{Cl}
	\subnumberinglabel{eq:splitSpace}
	\beta\mathbf A^T\mathbf y 
	= \beta \sum_{m=1}^M \mathbf a_m y_m\\
	\frac{\beta}{2}\mathbf A^T\mathbf A 
	= \frac{\beta}{2}\sum_{m=1}^M \mathbf a_m \mathbf a_m^T
\end{IEEEeqnarray}
where $\mathbf a_m=[\mathbf A^T]_{:, m} \in\mathbb C^{N \times 1}$ is the $m$-th column of $\mathbf A^T$, and $y_m$ is the $m$-th element of $\mathbf y$. 
Based on the splits in \eqref{eq:splitSpace}, we define $M$ auxiliary PDFs. The $m$-th auxiliary PDF is
\begin{IEEEeqnarray}{Cl}
	&p( \mathbf h;  {\bm\theta}_m, {\bm\Theta}_m)  
	= \exp\Big\{ -\frac{\beta}{2}\mathbf h^T{\bm\Lambda}_m\mathbf h 
	+ \beta\mathbf h^T{\bm\lambda}_m    - \frac{\beta}{2} \mathbf h^T\mathbf a_m\mathbf a_m^T\mathbf h 
	+ \beta\mathbf h^T \mathbf a_m y_m + \psi_m
	\Big\}.
	\IEEEeqnarraynumspace
\end{IEEEeqnarray}
where $\bm\Lambda_m$ is diagonal. Meanwhile, the extra PDF and the target PDF are defined as
\begin{IEEEeqnarray}{Cl}
	\label{eq:extrapd0}
	&\quad p_d( \mathbf h;  \hat{\bm\theta}_0, \hat{\bm\Theta}_0) = \exp\left\{-\beta\kappa\|\mathbf h\|_1 - \frac{\beta}{2}\mathbf h^T \hat{\bm\Lambda}_0\mathbf h 
	+ \beta\mathbf h^T \hat{\bm\lambda}_0 + \hat\psi_{0} \right\} 
\end{IEEEeqnarray}
and
\begin{IEEEeqnarray}{Cl}
	p( \mathbf h;  {\bm\theta}_0, {\bm\Theta}_0)
	&= \exp\left\{- \frac{\beta}{2}\mathbf h^T {\bm\Lambda}_0\mathbf h 
	+ \beta\mathbf h^T {\bm\lambda}_0 + \psi_{0} \right\}
\end{IEEEeqnarray}
where $\bm\Lambda_0$ and $\hat{\bm\Lambda}_0$ are diagonal matrices.

\subsection{Relation between the Extra PDF and the Target PDF}
Since $\hat{\bm\Lambda}_0$ is diagonal, we have that the elements of  $\mathbf h$ with the PDF $p_d( \mathbf h;  \hat{\bm\theta}_0, \hat{\bm\Theta}_0)$ are independent, and $p_d( \mathbf h;  \hat{\bm\theta}_0, \hat{\bm\Theta}_0)$ is the product of independent PDFs of $N$ elements of $\mathbf h$, \textit{i.e.},
\begin{IEEEeqnarray}{Cl}
	&\quad p_d(\mathbf h;\hat{\bm\theta}_0,\hat{\bm\Theta}_0) \propto \prod\limits_{n=1}^N\exp\left\{-\beta\kappa| h_n| - \frac{\beta}{2} [\hat{\bm\Lambda}_0]_n h_n^2
	+ \beta h_n [\hat{\bm\lambda}_0]_n \right\}.
\end{IEEEeqnarray}
To obtain the mean and variance of $p_d(\mathbf h;\hat{\bm\theta}_0,\hat{\bm\Theta}_0)$, we compute the mean and variance of each independent distribution.
For convenience, we define the following PDF
\begin{IEEEeqnarray}{Cl}
	p_{d}(h;\hat{\theta},\hat{\Theta}) 
	&= \frac{1}{Z_{p_{d}}} \exp\left\{-\beta \left(\kappa|h| + \frac{1}{2}\hat{\Lambda} h^2 - h\hat{\lambda} \right)\right\}
\end{IEEEeqnarray}
where $\hat{\theta} =  \beta\hat{\lambda}$,
	$\hat{\Theta} = -\frac{\beta}{2}\hat{\Lambda}$.
When $\beta \rightarrow  \infty$, the mean and variance of $p_{d}(h;\hat{\theta},\hat{\Theta})$ can be derived with the Laplace method of integration as in \cite{Rangan2012}. The form of the mean and variance have also been derived in  \cite{Donoho2010}.
	
	By extending the results of \cite{Donoho2010}, we have the following statement. 
	When $\beta \rightarrow  \infty$, $p_{d}(h;\hat{\theta},\hat{\Theta}) $ and the Gaussian PDF $p(h;\theta,\Theta)$ have the same mean and covariance, where
	$\theta = \beta\lambda$, 
	$\Theta = -\frac{\beta}{2}\Lambda$, 
	$\lambda = \upsilon(\hat\lambda,\hat\Lambda)$, 
	$\Lambda = \Upsilon(\hat\lambda,\hat\Lambda)$, 
	and the soft threshold functions $\upsilon(\hat\lambda,\hat\Lambda)$, $\Upsilon(\hat\lambda,\hat\Lambda)$ are defined as
	\begin{IEEEeqnarray}{Cl}
		\subnumberinglabel{eq:uUpsilons}
		\upsilon(\hat\lambda,\hat\Lambda) 
		&= 
		\begin{cases}
			(|\hat\lambda|-\kappa)  \text{sign}(\hat\lambda), \quad |\hat\lambda|>\kappa\\
			0, \quad \, \text{otherwise}
		\end{cases}\\
		\Upsilon(\hat\lambda,\hat\Lambda) 
		&=
		\begin{cases}
			\hat\Lambda, \quad |\hat\lambda|>\kappa\\
			\infty, \quad \, \text{otherwise}
		\end{cases}
	\end{IEEEeqnarray}
	The detailed derivation is provided in Appendix \ref{appendice: theorem Gaussian approx} for self-contained purpose.

From the above, we can find a target PDF $ p(\mathbf h;{\bm\theta}_0,{\bm\Theta}_0)$ with the same mean and variance as $p_d(\mathbf h;\hat{\bm\theta}_0,\hat{\bm\Theta}_0)$ when $\beta\rightarrow\infty$.
Then, we have 
\begin{IEEEeqnarray}{Cl}
	\subnumberinglabel{eq:SpaceGaApprom}
	\bm\mu_0
	= \hat{\bm\mu}_0 , \quad
	\bm\Sigma_0 
	= \hat{\bm\Sigma}_0 \\
	\bm\lambda_0
	= \upsilon(\hat{\bm\lambda}_0,\hat{\bm\Lambda}_0) , \quad
	\bm\Lambda_0
	= \Upsilon(\hat{\bm\lambda}_0,\hat{\bm\Lambda}_0).
\end{IEEEeqnarray}
Thus, $p(\mathbf h;{\bm\theta}_0,{\bm\Theta}_0) $ is a Gaussian approximation of $p_d(\mathbf h;\hat{\bm\theta}_0,\hat{\bm\Theta}_0) $.
The mean and covariance of $p(\mathbf h;{\bm\theta}_0,{\bm\Theta}_0) $ are obtained by
\begin{IEEEeqnarray}{ClCl}
	\subnumberinglabel{eq:SpaceGaAppro0}
	\bm\mu_0 
	&= {\bm\Lambda}_0^{-1} {\bm\lambda}_0, \quad
	&\bm\Sigma_0
	&= \frac{1}{\beta} {\bm\Lambda}_0^{-1}.
\end{IEEEeqnarray}

\subsection{Algorithm Derivation}

We derive the beliefs $\bm\xi_m$, $\bm\Xi_m$ corresponding to the $m$-th auxiliary point at the $t$-th iteration in the following theorem.

\begin{theorem}
	\label{th: theorem beliefs IGA}
	Define the beliefs $\bm\xi_m$, $\bm\Xi_m$ as $\bm\xi_m=\bm\lambda_m^0 - \bm\lambda_m$ and $\bm\Xi_m = \bm\Lambda_m^0 -\bm\Lambda_m$. Then, the beliefs are given by
	\begin{IEEEeqnarray}{Cl}
		\bm\xi_m  
		&= \left(r_m\mathbf I - \mathbf L_m{\bm\Lambda}_m^{-1} \right)^{-1}   \left((\mathbf L_m-\mathbf a_m\mathbf a_m^T){\bm\Lambda}_m^{-1}{\bm\lambda}_m + \mathbf a_m y_m\right) 
		\IEEEyesnumber\IEEEyessubnumber*
		\label{eq:IGAbeliefsxi}\\
		\bm\Xi_m
		&= \left(r_m\mathbf L_m^{-1}-{\bm\Lambda}_m^{-1}\right)^{-1}
		\label{eq:IGAbeliefsXi}
	\end{IEEEeqnarray}	
	where $r_m=1+\mathbf a_m^T{\bm\Lambda}_m^{-1}\mathbf a_m$ and $\mathbf L_m = \mathbf I\odot\mathbf a_m\mathbf a_m^T$.
\end{theorem}
\begin{proof}
	The proof is provided in Appendix \ref{appendice: theorem beliefs IGA}.
\end{proof}

After obtaining the beliefs $\bm\xi_m$, $\bm\Xi_m$, the parameters of the target and auxiliary points are updated as
\begin{IEEEeqnarray}{ClCl}
	\subnumberinglabel{eq:newlLambda0m}
	\hat{\bm\lambda}_0^{t+1}
	&= \sum_{m=1}^M \bm\xi_m^t  , \quad
	& \hat{\bm\Lambda}_0^{t+1} 
	&= \sum_{m=1}^M \bm\Xi_m^t  \\
	\bm\lambda_0^{t+1}  
	&=  \upsilon(\hat{\bm\lambda}_0^{t+1},
	\hat{\bm\Lambda}_0^{t+1}) 
	, \quad &
	\bm\Lambda_0^{t+1}  &=  \Upsilon(\hat{\bm\lambda}_0^{t+1},
	\hat{\bm\Lambda}_0^{t+1})\\
	\bm\lambda_m^{t+1}
	&= \bm\lambda_0^{t+1}
	- \bm\xi_m^t , \quad
	& \bm\Lambda_m^{t+1}
	&= \bm\Lambda_0^{t+1}
	- \bm\Xi_m^t.
	\label{eq:lLambdamtand1}
	\IEEEeqnarraynumspace
\end{IEEEeqnarray}
In practice, it might be necessary to use the damped update with damping coefficient $\alpha \in (0,1]$ to enhance the convergence of the algorithm without changing its fixed point. 
After the algorithm converges, the outputs are the mean $\bm\mu_0$ and variance $\bm\Sigma_0$ of $p(\mathbf h;\bm\theta_0,\bm\Theta_0)$.
The IGA for BPDN is summarized as Algorithm \ref{algo:SpaceIG}.


\begin{algorithm}
	\caption{IGA for BPDN}
	\label{algo:SpaceIG}
	\KwIn{The received signal $\mathbf y$, the regularization parameter $\kappa$ and the maximal iteration number $T$.}
	\KwOut{\textit{a posteriori} mean $\bm\mu_0^t$ of $\mathbf h$.}
	\textbf{Initialization}:
	Set $t=0$, 
	$\bm\lambda_0^t=\bm\lambda_n^t=\mathbf 0$,
	$\bm\Lambda_0^t=\bm\Lambda_n^t= \mathbf I$;\\
	\While{$t\leq T$}{
		Calculate $M$ beliefs as
		\begin{IEEEeqnarray}{Cl}
			\bm\xi_m^t
			&= \left(r_m^t\mathbf I - \mathbf L_m({\bm\Lambda}_m^t)\right)^{-1}   \left((\mathbf L_m-\mathbf a_m\mathbf a_m^T)({\bm\Lambda}_m^t)^{-1}{\bm\lambda}_m^t + \mathbf a_m y_m\right)\notag\\
			\bm\Xi_m^t
			&= \left(r_m^t\mathbf L_m^{-1}-({\bm\Lambda}_m^t)^{-1}\right)^{-1}\notag 
		\end{IEEEeqnarray}
		\\
		Update the parameters of the auxiliary, extra and target PDFs as \eqref{eq:newlLambda0m};\\
		$t=t+1$;
	}
	Output the \textit{a posteriori} mean of $\mathbf h$ as
	\begin{IEEEeqnarray}{Cl}
		\bm\mu_0^t
		&=  ({\bm\Lambda}_0^t)^{-1}{\bm\lambda}_0^t \notag
	\end{IEEEeqnarray}
\end{algorithm}

\section{Approximate Information Geometry Approach for BPDN}
\label{sec:AIGA}
In this section, we propose the AIGA for BPDN in standard linear regression. 
The derivations of the AIGA are given based on the IGA.

We assume that $a_{mn} = O(1/\sqrt M),\forall m,n$ in the following derivation as in the AMP algorithm. 
Recall the relationships between the parameters of auxiliary and target points
in \eqref{eq:lLambdamtand1}.
We use the approximations of $\bm\xi_m^t$ and $\bm\Xi_m^t$ to obtain the simpler AIGA. 

\subsection{Approximations of Beliefs}

We first derive the approximation of the belief $\bm\Xi_m^t$.
According to $r_m = 1+\mathbf a_m^T{\bm\Lambda}_m^{-1}\mathbf a_m = 1 + \sum_{n=1}^N a_{mn}^2 [\bm\Lambda_m^t]_{nn}^{-1}$,  $\mathbf L_m = \mathbf I\odot\mathbf a_m\mathbf a_m^T$ and $a_{mn} = O(1/\sqrt M)$, we have that $r_m \rightarrow 1 + \frac{1}{M} \sum_{n=1}^N [\bm\Lambda_m^t]_{nn}^{-1}$  
and $\mathbf L_m \rightarrow \frac{1}{M}  \mathbf I$.
Then, the elements of $\bm\Xi_m^t$ in \eqref{eq:IGAbeliefsXi} can be expressed as
\begin{IEEEeqnarray}{Cl}
	[\bm\Xi_m^t]_{nn}
	&= \left(r_m^t [\mathbf L_m]_{nn}^{-1} 
	- [\bm\Lambda_m^t]_{nn}^{-1}\right)^{-1} \notag\\
	&\rightarrow \left(M + \sum\limits_{n' \neq n} [\bm\Lambda_m^t]_{n'n'}^{-1}
	\right)^{-1}
 \notag\\
	&\rightarrow \left(M + \sum_{n'=1}^N [\bm\Lambda_m^t]_{n'n'}^{-1}
	\right)^{-1}
\end{IEEEeqnarray}
when $M$ and $N$ are large enough. 
From the above equation, we can see that  $[\bm\Xi_m^{t}]_{nn} = O(1/M)$. 
Then, we have ${\bm\Lambda}_m^{t+1}
= {\bm\Lambda}_0^{t+1}
- \bm\Xi_m^t \rightarrow {\bm\Lambda}_0^{t+1}$ and
\begin{IEEEeqnarray}{Cl}
	\label{eq:Xi0tapprox}
	[\bm\Xi_m^t]_{nn}
	\rightarrow \left(M + {\rm tr}\left(( \bm\Lambda_0^t)^{-1}\right)
	\right)^{-1}  
	\triangleq \Xi_0^t 
\end{IEEEeqnarray}
when $M$ and $N$ are large enough, where  
${\rm tr}(\cdot)$ is the trace of a matrix. 
This shows that the belief $\bm\Xi_m^t \rightarrow \Xi_0^t \mathbf I$  when $M$ and $N$ are sufficiently large. 


Then, we focus on the approximation of the belief $\bm\xi_m^t$. 
Define $\hat\Lambda_0^{t+1} = M \Xi_0^t$. 
From \eqref{eq:IGAbeliefsXi}, \eqref{eq:SpalLambdane2}, $\bm\Xi_m^t \rightarrow \Xi_0^t \mathbf I$ and $\mathbf L_m \rightarrow \frac{1}{M} \mathbf I$, we have
\begin{IEEEeqnarray}{Cl}
	\label{eq:xiformer}
	  \left(r_m^t\mathbf I - \mathbf L_m(\bm\Lambda_m^t)^{-1}\right)^{-1}  
	&= \left(r_m^t\mathbf L_m^{-1} 
	- (\bm\Lambda_m^t)^{-1}\right)^{-1} \mathbf L_m^{-1}\notag\\
	&= \bm\Xi_m^t \mathbf L_m^{-1}
	\rightarrow \hat\Lambda_0^{t+1} \mathbf I .
\end{IEEEeqnarray}
Then, we rewrite \eqref{eq:IGAbeliefsxi} as
\begin{IEEEeqnarray}{Cl}
	\label{eq:SpalLambdane2}
	\bm\xi_m^t
	&= \hat\Lambda_0^{t+1}\left((\frac{1}{M} \mathbf I-\mathbf a_m\mathbf a_m^T)(\bm\Lambda_m^t)^{-1} \bm\lambda_m^t + \mathbf a_m y_m\right).
\end{IEEEeqnarray}
From $\bm\lambda_m^{t+1}
	 = \bm\lambda_0^{t+1}
	- \bm\xi_m^t $ in \eqref{eq:lLambdamtand1} and $\bm\Lambda_m^{t+1} \rightarrow \bm\Lambda_0^{t+1}$, we get
\begin{IEEEeqnarray}{Cl}
	\label{eq:Zmt}
	\mathbf a_m^T ({\bm\Lambda}_m^t)^{-1} \bm\lambda_m^t 
	\rightarrow \mathbf a_m^T (\bm\Lambda_0^t)^{-1} ({\bm\lambda}_0^t - \bm\xi_m^{t-1})
	\triangleq Z_m^t
\end{IEEEeqnarray}
and
\begin{IEEEeqnarray}{Cl}
	\label{eq:DeltaZmt}
	 \frac{1}{M}(\bm\Lambda_m^t)^{-1} \bm\lambda_m^t 
	&\rightarrow \frac{1}{M} (\bm\Lambda_0^t)^{-1} (\bm\lambda_0^t - \bm\xi_m^{t-1}) \notag\\
	&\overset{(a)}{\rightarrow} \frac{1}{M} (\bm\Lambda_0^t)^{-1} \bm\lambda_0^t
\end{IEEEeqnarray}
where $\overset{(a)}{\rightarrow}$ is obtained because 
 we can extract the coefficient $a_{mn}$ from $ [\bm\xi_m^{t-1}]_n$ according to  \eqref{eq:SpalLambdane2}, and we have
$\frac{1}{M} [{\bm\Lambda}_0^t]_{nn}^{-1} [\bm\xi_m^{t-1}]_n 
\rightarrow 1/M^{3/2} \rightarrow 0$ when $M$ is large enough. 
Let $z_m^t = y_m - Z_m^t $ and combining the equations from \eqref{eq:SpalLambdane2} to \eqref{eq:DeltaZmt}, we have 
\begin{IEEEeqnarray}{Cl}
	\label{eq:xixi0delxi}
	\bm\xi_m^t
	&\rightarrow \hat\Lambda_0^{t+1} \Big( \mathbf a_m 
	( y_m - Z_m^t) + \frac{1}{M} (\bm\Lambda_0^t)^{-1} \bm\lambda_0^t \Big) 
\notag\\
	&= \hat{\Lambda}_0^{t+1} \Big( \mathbf a_m 
	  z_m^t + \frac{1}{M}  (\bm\Lambda_0^t)^{-1} \bm\lambda_0^t \Big).
\end{IEEEeqnarray}

\subsection{Updating Process with Approximations}
Next, we update $\hat{\bm\lambda}_0^{t+1}$ and $\hat{\bm\Lambda}_0^{t+1}$ based on the approximations of the beliefs. 
The parameter $\hat{\bm\lambda}_0^{t+1} = \sum_{m=1}^M \bm\xi_m^t$ is approximated as 
\begin{IEEEeqnarray}{Cl}
	\label{eq:lamb0appr}
	\hat{\bm\lambda}_0^{t+1} 
	&\rightarrow \hat{\Lambda}_0^{t+1} \left(\mathbf A^T \mathbf z^t
	+    
	(\bm\Lambda_0^t)^{-1} {\bm\lambda}_0^t\right)
\end{IEEEeqnarray}
where $[\mathbf z^t]_m = z_m^t$.
From $\hat{\bm\Lambda}_0^{t+1} 
	= \sum_{m=1}^M \bm\Xi_m^t \rightarrow M\Xi_0^t  \mathbf I $, we have $\hat{\bm\Lambda}_0^{t+1} \rightarrow \hat{\Lambda}_0^{t+1}  \mathbf I$. 
The natural parameters ${\bm\lambda}_0^{t+1},
{\bm\Lambda}_0^{t+1}$ of approximated Gaussian distributions 
 are then obtained by
\begin{IEEEeqnarray}{Cl}
	{\bm\lambda}_0^{t+1}
	&= \upsilon (\hat{\bm\lambda}_0^{t+1}, \hat{\Lambda}_0^{t+1}\mathbf{I})
	\IEEEyesnumber\IEEEyessubnumber*\\
	\bm\Lambda_0^{t+1}
	&= \Upsilon (\hat{\bm\lambda}_0^{t+1}, \hat{\Lambda}_0^{t+1}\mathbf{I}).
\end{IEEEeqnarray}
By using \eqref{eq:Zmt} and $z_m^t = y_m - Z_m^t $, we can express $z_m^t$ as 
\begin{IEEEeqnarray}{Cl}
	z_m^t 
	&= y_m - \mathbf a_m^T (\bm\Lambda_0^t)^{-1} \bm\lambda_0^t 
	+  \mathbf a_m^T (\bm\Lambda_0^t)^{-1} \bm\xi_m^{t-1}.
\end{IEEEeqnarray}
Substituting \eqref{eq:xixi0delxi}  into $\mathbf a_m^T (\bm\Lambda_0^t)^{-1} \bm\xi_m^{t-1}$ obtains 
\begin{IEEEeqnarray}{Cl}
	 \mathbf a_m^T (\bm\Lambda_0^t)^{-1}  \bm\xi_m^{t-1}  
	&= \mathbf a_m^T (\bm\Lambda_0^t)^{-1} \hat\Lambda_0^t \Big(\mathbf a_mz_m^{t-1} + \frac{1}{M} (\bm\Lambda_0^{t-1})^{-1}  \bm\lambda_0^{t-1} \Big) \notag\\
	&\overset{(b)}{\rightarrow} z_m^{t-1}\hat\Lambda_0^t  \sum_{n=1}^N a_{mn}^2
	[\bm\Lambda_0^t]_{nn}^{-1} \notag\\
	&\rightarrow \frac{1}{M} z_m^{t-1}\hat\Lambda_0^t  \sum_{n=1}^N 
	[\bm\Lambda_0^t]_{nn}^{-1}
	\notag\\
	&= \frac{1}{M} z_m^{t-1}\hat\Lambda_0^t  
	{\rm tr} \big((\bm\Lambda_0^t)^{-1} \big)
\end{IEEEeqnarray}
where $\overset{(b)}{\rightarrow}$ is because 
$\frac{1}{M} \hat\Lambda_0^t a_{mn}
[\bm\Lambda_0^t]_{nn}^{-1}
[\bm\Lambda_0^{t-1}]_{nn}^{-1}
[\bm\lambda_0^{t-1}]_n \rightarrow 1/M^{3/2} \rightarrow 0$.
Then $z_m^t$ is represented as
\begin{IEEEeqnarray}{Cl}
	z_m^t 
	&= y_m - \mathbf a_m^T (\bm\Lambda_0^t)^{-1} \bm\lambda_0^t
	+ \frac{1}{M} z_m^{t-1} \hat{\Lambda}_0^t 
	{\rm tr} \big((\bm\Lambda_0^t)^{-1} \big)
\end{IEEEeqnarray}
which means 
\begin{IEEEeqnarray}{Cl}
	\mathbf z^t
	&= \mathbf y - \mathbf A
	({\bm\Lambda}_0^t)^{-1} {\bm\lambda}_0^t
	+ \frac{1}{M} \mathbf z^{t-1}\hat{\Lambda}_0^t 
	{\rm tr} \big((\bm\Lambda_0^t)^{-1} \big).
\end{IEEEeqnarray}

Besides, from $\hat{\Lambda}_0^{t+1} = M \Xi_0^t$, $\bm\Lambda_m^{t+1} \rightarrow \bm\Lambda_0^{t+1}$ and \eqref{eq:Xi0tapprox}, we can get
\begin{IEEEeqnarray}{Cl}
	(\hat\Lambda_0^{t+1})^{-1}
	&= \frac{1}{M} (\Xi_0^t)^{-1}\notag\\
	&= 1 + \frac{1}{M} {\rm tr} \big((\bm\Lambda_0^t)^{-1} \big).
\end{IEEEeqnarray}
Then the resulting AIGA algorithm is written as
\begin{IEEEeqnarray}{Cl}
\hat{\bm\lambda}_0^{t+1}  
	&= \hat{\Lambda}_0^{t+1} (\mathbf A^T \mathbf z^t
	+   
	(\bm\Lambda_0^t)^{-1} {\bm\lambda}_0^t) \\
		\hat{\Lambda}_0^{t+1}
	&= \Big(1 + \frac{1}{M} {\rm tr} \big((\bm\Lambda_0^t)^{-1} \big) \Big)^{-1}
	\IEEEyesnumber\IEEEyessubnumber*
	\label{eq:IGLambda0up}\\
	\mathbf z^t
	&= \mathbf y - \mathbf A
	({\bm\Lambda}_0^t)^{-1} {\bm\lambda}_0^t 
	+ \frac{1}{M} \mathbf z^{t-1} \hat{\Lambda}_0^t 
	{\rm tr} \big((\bm\Lambda_0^t)^{-1} \big) 
	\label{eq:IGztup}\\
	{\bm\lambda}_0^{t+1}
	&= \upsilon (\hat{\bm\lambda}_0^{t+1}, \hat{\Lambda}_0^{t+1} )
	\label{eq:IGlamhatup}\\
	\bm\Lambda_0^{t+1}
	&= \Upsilon (\hat{\bm\lambda}_0^{t+1}, \hat{\Lambda}_0^{t+1}) .
	\label{eq:IGLamhatup}
	\IEEEeqnarraynumspace
\end{IEEEeqnarray}
Let $\iota^t$ be the ratio of elements in $\hat{\bm\lambda}_0^t$ whose absolute value is larger than $\kappa$, $\delta = M/N$, then $\frac{1}{M} {\rm tr} \big((\bm\Lambda_0^t)^{-1} \big) = \frac{\iota^t}{\delta} (\hat\Lambda_0^t)^{-1}$.
The AIGA for BPDN is summarized as Algorithm \ref{algo:SpaceAIG}.
\begin{algorithm}
	\caption{AIGA for BPDN}
	\label{algo:SpaceAIG}
	\KwIn{The received signal $\mathbf y$, the regularization parameter $\kappa$ and the maximal iteration number $T$.}
	\KwOut{\textit{a posteriori} mean $\bm\mu_0^t$ of $\mathbf h$.}
	\textbf{Initialization}:
	Set $t=0$, 
	$\bm\lambda_0^t = \mathbf 0$, 
	$\bm\Lambda_0^t = \infty$, 
	$\mathbf z^{t-1} = \mathbf y$, 
	$\hat\Lambda_0^t = 1$, 
	$\iota^t = 0$;\\
	\While{$t\leq T$}{
		Update $\hat\Lambda_0^{t+1} = \left(1+\frac{\iota^t}{\delta} (\hat\Lambda_0^t)^{-1}\right)^{-1} $;\\
		Update $\mathbf z^t
		= \mathbf y - \mathbf A^T
		(\bm\Lambda_0^t)^{-1} \bm\lambda_0^t 
		+ \frac{\iota^t}{\delta} \mathbf z^{t-1}$;\\
		Calculate $\hat{\bm\lambda}_0^{t+1} = \hat\Lambda_0^{t+1}( \mathbf A^T \mathbf z^t
		+      
		(\bm\Lambda_0^t)^{-1} \bm\lambda_0^t)$;\\
		Denote the ratio of elements in $\hat{\bm\lambda}_0^{t+1}$ whose absolute value is larger than $\kappa$ as $\iota^{t+1}$;\\
		Calculate $\bm\lambda_0^{t+1}$ and $\bm\Lambda_0^{t+1}$: 
		 If $|[\hat{\bm\lambda}_0^{t+1}]_n| > \kappa$  
		 then  
		 $[\bm\lambda_0^{t+1}]_n = \left(|[\hat{\bm\lambda}_0^{t+1}]_n| - \kappa \right)  \text{sign} \left(|[\hat{\bm\lambda}_0^{t+1}]_n|\right)$, 
		$[\bm\Lambda_0^{t+1}]_n = \hat\Lambda_0^{t+1}$; 
		 else  
		$[\bm\lambda_0^{t+1}]_n = 0$, $[\bm\Lambda_0^{t+1}]_n = \infty$;\\
		$t=t+1$;
	}
	When convergence or $t>T$, calculate the mean of the target point as
	\begin{IEEEeqnarray}{Cl}
		\bm\mu_0^t
		&=  (\bm\Lambda_0^t)^{-1} \bm\lambda_0^t .\notag
	\end{IEEEeqnarray}
\end{algorithm}

 \begin{figure}[htbp] 
	\centering
	\includegraphics[width=0.4\linewidth]{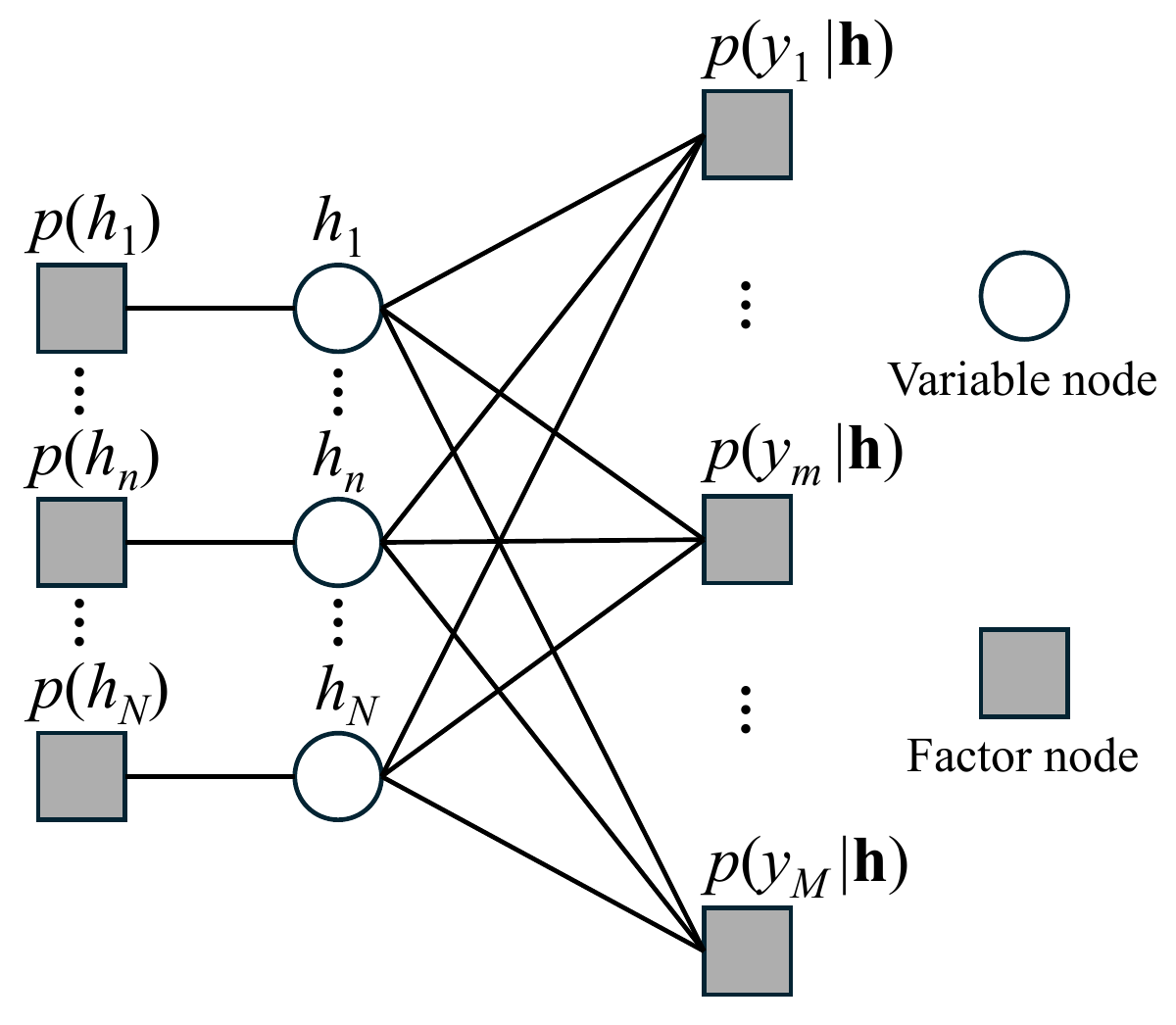}
	\caption{Factor graph of \eqref{eq:postupost}}
	\label{fig:factorGch4}
\end{figure}

\section{The equivalence of AMP and AIGA}
\label{sec:AMPeqIGA}
In this section, we give proof of the equivalence of the AMP algorithm and the AIGA.
Specifically, we present the relationship between the MP algorithm and the IGA, and the equivalence of the AMP algorithm and the AIGA.

\subsection{Review of MP and IGA}
We first review the MP algorithm. Based on Fig. \ref{fig:factorGch4}, the postulated \textit{a posteriori} PDF of $\mathbf h$ is represented as
\begin{IEEEeqnarray}{Cl}
	p( \mathbf{h}| \mathbf{y}) 
	&= \frac{1}{Z_p} 
	\prod\limits_{n=1}^N p_n(h_n)
	\prod\limits_{q=1}^Q \varphi_q(\mathbf h)\notag\\
	&\propto \prod\limits_{n=1}^N \exp\left\{-\beta\kappa|h_n|\right\}
	\prod\limits_{m=1}^M p(y_m|\mathbf h).
\end{IEEEeqnarray}
Then the messages in MP are addressed as
\begin{IEEEeqnarray}{Cl}
	\nu_{n\rightarrow m}^{t+1}(h_n)
	&\propto \exp\left\{-\beta\kappa|h_n|\right\}
	\prod_{m'\neq m}^M 
	\nu_{m'\rightarrow n}^t (h_n)
	\IEEEyesnumber\IEEEyessubnumber*
	\label{eq:AMPmess1}\\
	\nu_{m\rightarrow n}^t(h_n)
	&\propto \int p(y_m|\mathbf h)
	\prod\limits_{n'\neq n}^N 
	\nu_{n'\rightarrow m}^t(h_n) \,
	\mathrm d \mathbf h_{\backslash n}
	\label{eq:AMPmess2}
\end{IEEEeqnarray}
where $\mathbf h_{\backslash n}$ is $\mathbf h$ except $h_n$, 
$\nu_{n\rightarrow m}^{t+1}(h_n)$ is the message from the variable node $h_n$ to the factor node $p(y_m|\mathbf h)$, and
$\nu_{m\rightarrow n}^t(h_n)$ is the message in the opposite direction. 
Denote their means and variances as $\mu_{n\rightarrow m}^{t+1}$, $\tau_{n\rightarrow m}^{t+1}$
and $\mu_{m\rightarrow n}^t$, $\tau_{m\rightarrow n}^t$, respectively.
When $\beta \rightarrow \infty$, $\mu_{n\rightarrow m}^{t+1}$ and $\tau_{n\rightarrow m}^{t+1}$ are obtained by 
\begin{IEEEeqnarray}{Cl}
	\mu_{n\rightarrow m}^{t+1}
	&= \eta\left(r_{n\rightarrow m}^t, \Sigma_{n\rightarrow m}^t\right) 
	\IEEEyesnumber\IEEEyessubnumber*\\
	\tau_{n\rightarrow m}^{t+1}
	&= \Sigma_{n\rightarrow m}^t \eta'\left(r_{n\rightarrow m}^t, \Sigma_{n\rightarrow m}^t\right)
\end{IEEEeqnarray}
where
\begin{IEEEeqnarray}{Cl}
	r_{n\rightarrow m}^t
	&= \Sigma_{n \rightarrow m}^t \sum_{m' \neq m} \left(\tau_{m'\rightarrow n}^t \right)^{-1} \mu_{m' \rightarrow n}^t
	\IEEEyesnumber\IEEEyessubnumber*\\
	\Sigma_{n\rightarrow m}^t
	&= \left( \sum_{m' \neq m} \left(\tau_{m'\rightarrow n}^t \right)^{-1}\right)^{-1}
\end{IEEEeqnarray}
and $\eta(\mu,\tau)$, $\eta'(\mu,\tau)$ are defined as
\begin{IEEEeqnarray}{Cl}
	\eta(\mu,\tau) 
	&=
	\begin{cases}
		(|\mu|-\tau)  \text{sign}(\mu), \quad |\mu|>\tau\\
		0, \quad \, \text{otherwise}
	\end{cases}
	\IEEEyesnumber\IEEEyessubnumber*\\
	\eta'(\mu,\tau)
	&=
	\begin{cases}
		1, \quad |\mu|>\tau\\
		0, \quad \, \text{otherwise}.
	\end{cases}
\end{IEEEeqnarray}

Next, we review the IGA. 
Rewrite the auxiliary PDF as 
\begin{IEEEeqnarray}{Cl}
	p( \mathbf h; \bm\theta_m^t, \bm\Theta_m^t)
	&\propto p(\mathbf h;\tilde{\bm\theta}_m^{t},\tilde{\bm\Theta}_m^{t})p(y_m|\mathbf h)\notag\\ 
	&\propto \exp\bigg\{ -\frac{\beta}{2}\mathbf h^T \bm\Lambda_m^t\mathbf h 
	+ \beta\mathbf h^T \bm\lambda_m^t  - \frac{\beta}{2} \mathbf h^T\mathbf a_m\mathbf a_m^T\mathbf h 
	+ \beta\mathbf h^T \mathbf a_m y_m \bigg\}
\end{IEEEeqnarray}
where $p(\mathbf h;\tilde{\bm\theta}_m^{t},\tilde{\bm\Theta}_m^{t})$ is represented as 
\begin{IEEEeqnarray}{Cl}
	\label{eq:IGmess1}
	p(\mathbf h;\tilde{\bm\theta}_m^{t},\tilde{\bm\Theta}_m^{t})
	&\propto \exp\left\{-\frac{\beta}{2}\mathbf h^T\bm\Lambda_m^t\mathbf h 
	+ \beta \mathbf h^T \bm\lambda_m^t
	\right\}
	\IEEEeqnarraynumspace
\end{IEEEeqnarray}
and the elements of $\bm\lambda_m^t$ and $\bm\Lambda_m^t$ are the messages from variable nodes $h_1,\cdots,h_N$ to the factor node $y_m$. 

The $m$-projection of auxiliary point to target manifold is
\begin{IEEEeqnarray}{Cl}
	  p\left(\mathbf h;(\bm\theta_m^0)^t,(\bm\Theta_m^0)^t\right)  
	&\propto \exp\left\{-\frac{\beta}{2}\mathbf h^T(\bm\Lambda_m^0)^t\mathbf h 
	+ \beta\mathbf h^T (\bm\lambda_m^0)^t 
	\right\}\notag\\
	&= \exp\left\{-\frac{\beta}{2}\mathbf h^T\bm\Lambda_m^t\mathbf h 
	+ \beta\mathbf h^T \bm\lambda_m^t 
	-\frac{\beta}{2} \mathbf h^T \bm\Xi_m^t \mathbf h 
	+ \beta \mathbf h^T \bm\xi_m^t
	\right\}\notag\\
	&\propto p(\mathbf h;\tilde{\bm\theta}_m^{t},\tilde{\bm\Theta}_m^{t}) p(\mathbf h;\bm\xi_m^t,\bm\Xi_m^t)
\end{IEEEeqnarray}
where $p(\mathbf h;\bm\xi_m^t,\bm\Xi_m^t)$ is represented as
\begin{IEEEeqnarray}{Cl}
	\label{eq:IGmess2}
	p(\mathbf h;\bm\xi_m^t,\bm\Xi_m^t)
	\propto \exp\left\{-\frac{\beta}{2}\mathbf h^T\bm\Xi_m^t\mathbf h 
	+ \beta\mathbf h^T \bm\xi_m^t
	\right\}
\end{IEEEeqnarray}
which is the approximation of the likelihood function of $y_m$. 
The elements of $\bm\xi_m^t$ and $\bm\Xi_m^t$ are beliefs from the factor node $y_m$ to variable nodes $h_1,\cdots,h_N$.

\subsection{The relationship between MP and IGA}
\label{subsec:MP_IGA}
In this subsection, we present the relationship between the MP algorithm  and the IGA.

First, we analyze the relationship between $\nu_{m\rightarrow n}^t(h_n)$ and $p(\mathbf h;\bm\xi_m^t,\bm\Xi_m^t)$.
We have the following intermediate variables in the MP algorithm \cite{Donoho2010} 
\begin{IEEEeqnarray}{Cl}
	Z_{m\rightarrow n} ^t
	&= \sum\limits_{n'\neq n} a_{mn'}
	\mu_{n'\rightarrow m}^t
	\IEEEyesnumber\IEEEyessubnumber*\\
	T_{m \rightarrow n}^t
	&= \sum\limits_{n'\neq n} a_{mn'}^2
	\tau_{n'\rightarrow m}^t .
	\label{eq:Tmton}
\end{IEEEeqnarray}
According to \cite[Equations (5), (6)]{Donoho2010}, the messages from the factor node to the variable node in the MP algorithm  are
\begin{IEEEeqnarray}{Cl}
	\mu_{m\rightarrow n}^t
	&= \frac{y_m - Z_{m\rightarrow n}^t}
	{a_{mn}}
	\IEEEyesnumber\IEEEyessubnumber*\\
	\tau_{m\rightarrow n}^t
	&= \frac{1 + T_{m \rightarrow n}^t}
	{\beta a_{mn}^2} .
\end{IEEEeqnarray}

Then we focus on the beliefs of the IGA. Recall that  $r_m^t = 1+\mathbf a_m^T\bm\Lambda_m^{-1}\mathbf a_m = 1+\sum_n a_{mn}^2[\bm\Lambda_m^t]_{nn}^{-1}$, and $[\mathbf L_m]_{nn} = [\mathbf I \odot \mathbf a_m \mathbf a_m^T]_{nn}= a_{mn}^2$. 
The $n$-th diagonal element of the variance of $p(\mathbf h;\bm\xi_m^t,\bm\Xi_m^t)$ is 
\begin{IEEEeqnarray}{Cl}
	\label{eq:Xiinveqtaumn}
	\frac{1}{\beta}[\bm\Xi_m^t]_{nn}^{-1}
	&= \frac{1}{\beta}\left[r_m^t\mathbf L_m^{-1} - (\bm\Lambda_m^t)^{-1}\right]_{nn}\notag\\
	&= \frac{1 + \sum_{n'\neq n} a_{mn'}^2 [\bm\Lambda_m^t]_{n'n'}^{-1}}
	{\beta a_{mn}^2} .
\end{IEEEeqnarray}
Besides, the $n$-th element of parameter $\bm\xi_m^t$ in $p(\mathbf h;\bm\xi_m^t,\bm\Xi_m^t)$ can be simplified as
\begin{IEEEeqnarray}{Cl}
	\big[\bm\xi_m^t\big]_n
	&= \Big[\Big(r_m^t\mathbf I - \mathbf L_m (\bm\Lambda_m^t)^{-1}\Big)^{-1}    \left((\mathbf L_m-\mathbf a_m\mathbf a_m^T) (\bm\Lambda_m^t)^{-1} \bm\lambda_m^t + \mathbf a_m y_m\right)\Big]_n\notag\\
	&= \frac{a_{mn} \Big(y_m
		- \sum_{n'\neq n} a_{mn'}[\bm\Lambda_m^t]_{n'n'}^{-1} [\bm\lambda_m^t]_{n'} \Big) }
	{1 + \sum_{n'\neq n} a_{mn'}^2[\bm\Lambda_m^t]_{n'n'}^{-1}}.
\end{IEEEeqnarray}
Then, the $n$-th element of the mean of $p(\mathbf h;\bm\xi_m^t,\bm\Xi_m^t)$ is 
\begin{IEEEeqnarray}{Cl}
	[\bm\Xi_m^t]_{nn}^{-1}[\bm\xi_m^t]_n
	&= \frac{ y_m
		- \sum_{n'\neq n} a_{mn'} [\bm\Lambda_m^t]_{n'n'}^{-1} [\bm\lambda_m^t]_{n'} }
	{a_{mn}} .
\end{IEEEeqnarray}
It can be observed that $\mu_{n\rightarrow m}^t$, $\tau_{n\rightarrow m}^t$ and  $[\bm\lambda_m^t]_n$, $[\bm\Lambda_m^t]_{nn}^{-1}$ are the messages in $\nu_{n\rightarrow m}^{t+1}(h_n)$ and $p(\mathbf h;\tilde{\bm\theta}_m^{t},\tilde{\bm\Theta}_m^{t})$, respectively.
The parameters $[\bm\Xi_m^t]_{nn}^{-1}[\bm\xi_m^t]_n$ and $\frac{1}{\beta} [\bm\Xi_m^t]_{nn}^{-1}$ of $p(\mathbf h;\bm\xi_m^t,\bm\Xi_m^t)$ in the IGA, has the same form as $\mu_{m\rightarrow n}^t$ and $\tau_{m\rightarrow n}^t$ of $\nu_{m\rightarrow n}^t(h_n)$ in the MP algorithm, respectively. 
However, the mean and variance of $\nu_{n\rightarrow m}^{t+1}(h_n)$ and $p(\mathbf h;\tilde{\bm\theta}_m^{t},\tilde{\bm\Theta}_m^{t})$ are not exactly equal, where $\mu_{n\rightarrow m}^t$ and $\tau_{n\rightarrow m}^t$ are obtained by calculating the expectations of the non-Gaussian PDF $\nu_{n\rightarrow m}^{t+1}(h_n)$, while $[\bm\lambda_m^t]_{n}$ and $[\bm\Lambda_m^t]_{nn}^{-1}$ are obtained by \eqref{eq:lLambdamtand1}. Thus, the new IG approach is simpler than the MP algorithm.
If we set up an extra PDF like \eqref{eq:extrapd0} for the independent terms of every auxiliary PDF, and calculate $\bm\lambda_m^t$ and $\bm\Lambda_m^t$ as the method in \eqref{eq:uUpsilons}, then this new IG approach is exactly equivalent to the MP algorithm.




We present the correspondence of parameters in the MP algorithm and the IGA in Table \ref{tab:MPeqIG}.
\begin{table}[htbp]
	\centering
	\caption{The correspondence of parameters in MP and IGA}
	\label{tab:MPeqIG}
	\begin{tabular}{cc}
		\hline
		MP & IGA\\
		\hline
		$\prod_n\nu_{n\rightarrow m}(h_n)$ 
		& $p(\mathbf h;\tilde{\bm\theta}_m^{t},\tilde{\bm\Theta}_m^{t})$ \\
		$\mu_{n\rightarrow m}$ 
		& $[{\bm\Lambda}_m]_{nn}^{-1} [{\bm\lambda}_m]_n $\\
		$\frac{1}{\beta}\tau_{n\rightarrow m}$ 
		& $\frac{1}{\beta}[{\bm\Lambda}_m]_{nn}^{-1}$\\
		$\prod_n\nu_{m\rightarrow n}(h_n)$ 
		& $p(\mathbf h;\bm\xi_m^t,\bm\Xi_m^t)$\\
		$\mu_{m\rightarrow n}$ 
		& $[\bm\Xi_m]_{nn}^{-1}[\bm\xi_m]_n$\\
		$\tau_{m\rightarrow n}$ 
		& $\frac{1}{\beta}[\bm\Xi_m]_{nn}^{-1}$\\
		\hline
	\end{tabular}
\end{table}

\subsection{Equivalence of AMP and AIGA}
In this subsection, we present the equivalence of the AMP algorithm and the AIGA.
The AMP algorithm is 
\begin{IEEEeqnarray}{Cl}
	\gamma^t 
	&= \frac{\kappa + \gamma^{t-1}}{\delta} 
	\left\langle \eta' (\mathbf A^T \mathbf z^{t-1} + \bm\mu^{t-1},\kappa + \gamma^{t-1} ) \right\rangle 
	\IEEEyesnumber\IEEEyessubnumber*
	\label{eq:AMPgammaup}\\
	\mathbf z^t
	&= \mathbf y - \mathbf A \bm\mu^t \notag\\
	&\quad 
	+ \frac{1}{\delta} \mathbf z^{t-1} \left\langle \eta'(\mathbf A^T \mathbf z^{t-1} + \bm\mu^{t-1},\kappa + \gamma^{t-1}) \right\rangle
	\label{eq:AMPztup}\\
	\bm\mu^{t+1} 
	&= \eta (\mathbf A^T \mathbf z^t + \bm\mu^t,\kappa + \gamma^t )
	\label{eq:AMPmutup}
	\IEEEeqnarraynumspace
\end{IEEEeqnarray}
where $\langle \cdot \rangle$ is the average of all element of a vector, $\gamma^t = \kappa \tau_z^t$, and \eqref{eq:AMPgammaup} is obtained by organizing the following equation
\begin{IEEEeqnarray}{Cl}
	\label{eq:AMPtauzup}
	\tau_z^t 
	&= \frac{1 + \tau_z^{t-1}}{\delta} 
	\left\langle \eta'(\mathbf A^T \mathbf z^{t-1} + \bm\mu^{t-1},\kappa(1 + \tau_z^{t-1})) \right\rangle .
\end{IEEEeqnarray}

The equivalence of the AMP algorithm and the AIGA is given in the following theorem.
\begin{theorem}
	\label{th: theorem AMP eq AIGA}
	The AMP algorithm and the AIGA are equivalent. 
	Specifically, 
	\eqref{eq:AMPtauzup} and \eqref{eq:IGLambda0up} are equivalent, 
	\eqref{eq:AMPztup} and \eqref{eq:IGztup} are equivalent, 
	\eqref{eq:AMPmutup} and \eqref{eq:IGlamhatup}, \eqref{eq:IGLamhatup} are equivalent.
\end{theorem}
\begin{proof}
	The proof is provided in Appendix \ref{appendice: theorem AMP eq AIGA}.
\end{proof}
	
We present the equivalent parameters of the AMP algorithm and the AIGA in Table \ref{tab:AMPeqAIG}.
\begin{table}[htbp]
	\centering
	\caption{The equivalence of parameters of AMP and AIGA}
	\label{tab:AMPeqAIG}
	\begin{tabular}{cc}
		\hline
		AMP & AIGA\\
		\hline
		$1+\tau_z^t$ & $\big(\hat\Lambda_0^{t+1}\big)^{-1}$\\
		$\mathbf z_{\text{AMP}}^t$ & $\mathbf z_{\text{AIGA}}^t$\\
		$\bm\mu^{t+1}$  &  $(\bm\Lambda_0^{t+1})^{-1} \bm\lambda_0^{t+1}$ \\
		\hline
	\end{tabular}
\end{table}

\section{Simulation Results}
\label{sec:Sim}

The normalized mean-squared error (NMSE) is used as the performance metric for the estimation, which is defined as
\begin{IEEEeqnarray}{Cl}
	\label{eq:NMSE}
	\text{NMSE} = 
	\frac{1}{L} \sum\limits_{l=1}^L
	\frac{\|\bar{\mathbf h}_l - \mathbf h_l\|_2^2}
	{\|\mathbf h_l\|_2^2}
\end{IEEEeqnarray}
where $\mathbf h_l$ is the original signal, and $\bar{\mathbf h}_l$ is the estimated signal, $L$ is the number of samples. We set $L=100$ in the simulations. 

\begin{figure}[htbp] 
	\centering
	\includegraphics[width=0.75\linewidth]{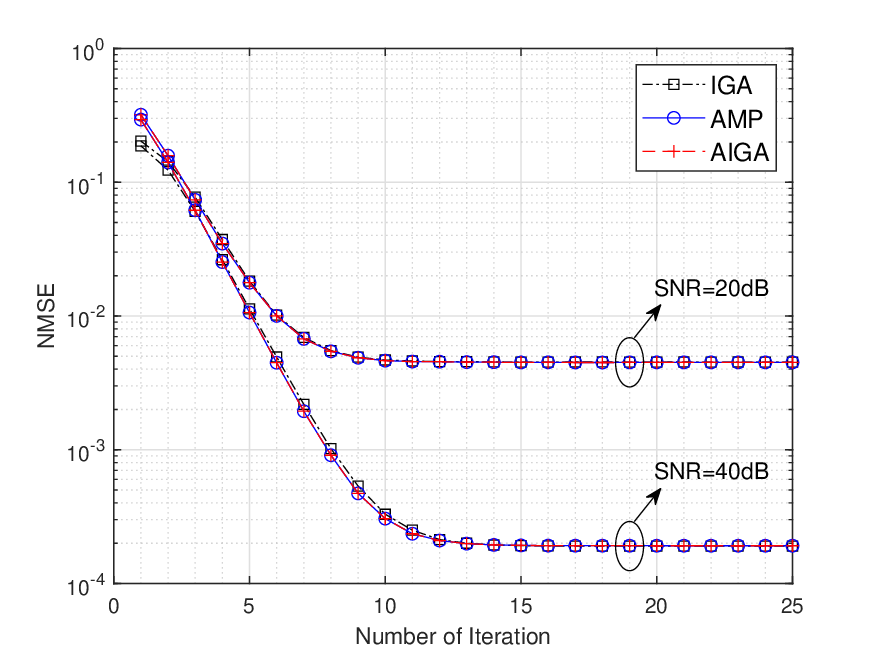}
	\caption{Convergence performance of IGA, AMP and AIGA, where $M=512$, $N=1024$, $\rho = 0.05$, $\text{SNR}=\{20,40\}\text{dB}$ and $\kappa=0.05$.}
	\label{fig:equality_AIGA_Real}
\end{figure}
\begin{figure}[htbp] 
	\centering
	\includegraphics[width=0.75\linewidth]{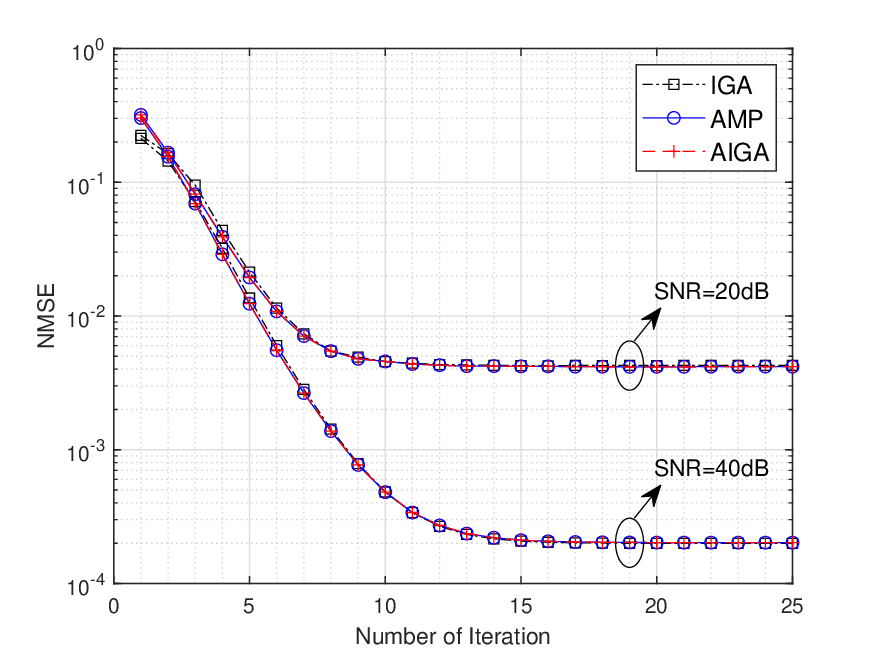}
	\caption{Convergence performance of IGA, AMP and AIGA, where $M=256$, $N=512$, $\rho = 0.05$, $\text{SNR}=\{20,40\}\text{dB}$ and $\kappa=0.05$.}
	\label{fig:equality_AIGA_Real2}
\end{figure}

We set the measurement matrix $\mathbf A \in \mathbb R^{M \times N}$ to an \textit{i.i.d.} Gaussian matrix with $[\mathbf A]_{mn}\sim\mathcal N(0,1/M),\forall m,n$. 
The signal $\mathbf h \in \mathbb R^{N \times 1}$ has \textit{i.i.d.} elements, where $[\mathbf h]_n \sim\mathcal {BG}(0,\rho),\forall n$. 
Figs.~\ref{fig:equality_AIGA_Real} and \ref{fig:equality_AIGA_Real2} show the convergence performance of the IGA, the AMP algorithm and the AIGA. 
The update of beliefs or messages may need damping to enhance convergence, but the damping coefficients for the update are different in three algorithms. 
Thus, to fairly compare these three algorithms, we do not use the damping coefficinets to update the parameters  in the simulations. 
Simulation results show that the NMSE curves of the AMP algorithm and the AIGA coincide exactly,
which verifies the equivalence of the AMP algorithm and the AIGA.

\section{Conclusion}
\label{sec:Con}
In this paper, a general information geometry framework for non-Gaussian prior standard linear regression problems is proposed. The proposed framework offers an intuitive and unified view to understand the standard linear regression problem.  The IGA and AIGA are then derived based on the proposed framework.  
We show that the IGA is strongly related to the MP algorithm, and prove that the AIGA is equivalent to the AMP algorithm. 
These results provide a new perspective and understanding for the AMP algorithm and open up new possibilities for improving stochastic reasoning methods.

\appendices

\section{Detailed Derivation of Gaussian Approximation of $p_{d}(h;\hat{\theta},\hat{\Theta}) $}
\label{appendice: theorem Gaussian approx}
According to \cite{Rangan2012}, when $\beta \rightarrow  \infty$, the mean $\mu_{p_{d}}$ of $p_{d}(h;\hat\theta,\hat\Theta)$ can be calculated with Laplace method of integration as
\begin{IEEEeqnarray}{Cl}
	 \lim\limits_{\beta \rightarrow \infty}
	 \mathbb E_{p_{d}}\{h\}  
	&=\lim\limits_{\beta \rightarrow \infty}\int h \frac{1}{Z} \exp\left\{-\beta\left(\kappa|h| + \frac{1}{2}\hat\Lambda h^2-h\hat\lambda\right)\right\} \, \mathrm d h\notag\\
	&=\mathop{\arg\min}\limits_h \left\{ \kappa|h|+ \frac{1}{2}\hat\Lambda h^2-h\hat\lambda\right\}.
\end{IEEEeqnarray}
Rewrite the absolute value $|h|$ as $(h^2)^{\frac{1}{2}}$, and define $g(h)$ as
\begin{IEEEeqnarray}{Cl}
	g(h)
	=\kappa(h^2)^{\frac{1}{2}} + \frac{1}{2}\hat\Lambda h^2-h\hat\lambda.
\end{IEEEeqnarray}
When $h=0$, we have that $g(h)=0$. When $h \neq 0$, the derivative of $g(h)$ is
\begin{IEEEeqnarray}{Cl}
	\frac{\partial g(h)}{\partial h}
	&= \hat\Lambda h-\hat\lambda + \frac{\kappa}{2}(h^2)^{-\frac{1}{2}}  2h\notag\\
	&= \hat\Lambda h - \hat\lambda + \kappa  \text{sign}(h) .
\end{IEEEeqnarray}
We have $\hat\Lambda > 0$ because the covariance $\hat\Sigma$ of $p(h;\hat\theta,\hat\Theta)$ satisfies $\hat\Sigma = -\frac{1}{2}\hat\Theta^{-1} = \frac{1}{\beta}\hat\Lambda^{-1} > 0$.
Since $\kappa >0$ and $\hat\Lambda > 0$, $g(h)$ is the sum of two convex functions and has a unique minimum point.
When $h \neq 0$, the zero of $\frac{\partial g(h)}{\partial h}$ satisfies
\begin{IEEEeqnarray}{Cl}
	\hat\lambda = |\hat\lambda|  \text{sign}(\hat\lambda)
	&= \hat\Lambda h + \kappa  \text{sign}(h) \notag\\
	&= (\hat\Lambda |h|+\kappa)  \text{sign}(h)
\end{IEEEeqnarray}
and we further have 
\begin{IEEEeqnarray}{Cl}
	|h|&= \hat\Lambda^{-1}(|\hat\lambda|-\kappa)\notag\\
	\text{sign}(h)&=\text{sign}(\hat\lambda).
\end{IEEEeqnarray}
Define the soft threshold function $\zeta(\hat\lambda,\hat\Lambda)$ as
\begin{IEEEeqnarray}{Cl}
	\zeta(\hat\lambda,\hat\Lambda) =
	\begin{cases}
		\hat\Lambda^{-1}(|\hat\lambda|-\kappa)  \text{sign}(\hat\lambda), \quad |\hat\lambda|>\kappa\\
		0, \quad \, \text{otherwise}
	\end{cases}
\end{IEEEeqnarray}
then the mean of  $p_{d}(h;\hat\theta,\hat\Theta)$ satisfies $\mu_{p_{d}} 
= \zeta(\hat\lambda,\hat\Lambda)$. 
Besides, the derivative of $\zeta(\hat\lambda,\hat\Lambda)$ \textit{w.r.t} $\hat\lambda$ is
\begin{IEEEeqnarray}{Cl}
	\zeta'(\hat\lambda,\hat\Lambda) =
	\begin{cases}
		\hat\Lambda^{-1}, \quad |\hat\lambda|>\kappa\\
		0, \quad \, \text{otherwise}
	\end{cases}
\end{IEEEeqnarray}
which is related to the variance $\Sigma_{p_{d}}$ of $p_{d}(h;\hat\theta,\hat\Theta)$. 

Next we calculate the variance of $p_{d}(h;\hat\theta,\hat\Theta)$. 
	\label{lm: lemma var of P}
	Let $q(h)$ be an arbitrary bounded and non-negative function, and define a PDF
	\begin{IEEEeqnarray}{Cl}
		\mathcal P(h) &= \frac{q(h) p(h;\hat\theta,\hat\Theta)}{\int q(h)p(h;\hat\theta,\hat\Theta)\, \, \mathrm d h}
	\end{IEEEeqnarray}
	which is a product of $q(h)$ and a Gaussian PDF $p(h;\hat\theta,\hat\Theta)$. 
	Its mean and variance are denoted as
	\begin{IEEEeqnarray}{Cl}
		\mathbb E_{\mathcal P}\{h\}&=\int h\mathcal P(h)\, \mathrm d h
		\IEEEyesnumber\IEEEyessubnumber*\\
		\mathbb V_{\mathcal P}\{h\}&=\int h^2\mathcal P(h)\, \mathrm d h - \mathbb E_{\mathcal P}\{h\}^2 .
	\end{IEEEeqnarray}
By using the calculus of variations \cite{gelfand2000calculus}, the derivative of $\mathbb E_{\mathcal P}\{h\}$ \textit{\textit{w.r.t}.} $\hat\theta$ is computed by 
\begin{IEEEeqnarray}{Cl}
	\frac{\partial \mathbb E_{\mathcal P}\{h\}}{\partial \hat\theta}&=\frac{\partial}{\partial \hat\theta}\frac{\int h q(h)p(h;\hat\theta,\hat\Theta)\, \mathrm d h}{\int q(h)p(h;\hat\theta,\hat\Theta)\, \mathrm d h}\IEEEnonumber*\\
	&=\frac{\int h^2q(h)p(h;\hat\theta,\hat\Theta)\, \mathrm d h \int q(h)p(h;\hat\theta,\hat\Theta)\, \mathrm d h}  {(\int q(h)p(h;\hat\theta,\hat\Theta)\, \mathrm d h)^2}\\
	&\ \ -\frac{\int h q(h)p_{\mathcal N}(h;\hat\theta,\hat\Theta)\, \mathrm d h \int hq(h)p(h;\hat\theta,\hat\Theta)\, \mathrm d h}  {(\int q(h)p(h;\hat\theta,\hat\Theta)\, \mathrm d h)^2}\\
	&=\int h^2\mathcal P(h)\, \mathrm d h
	-\left(\int h\mathcal P(h)\, \mathrm d h\right)^2\\
	&=\mathbb V_{\mathcal P}\{h\}\IEEEyesnumber
\end{IEEEeqnarray}
which means $\mathbb V_{\mathcal P}\{h\} = \frac{\partial \mathbb E_{\mathcal P}\{h\}}{\partial \hat\theta}$. 


Since $p_{d}(h;\hat\theta,\hat\Theta) \propto \exp\left\{-\beta \kappa|h|\right\}p(h;\hat\theta,\hat\Theta)$ and $\hat\theta=\beta\hat\lambda$,
we have
\begin{IEEEeqnarray}{Cl}
	\frac{\partial \mathbb E_{p_{d}}\{h\}}{\partial \hat\lambda} 
	= \beta\frac{\partial \mathbb E_{p_{d}}\{h\}}{\partial \hat\theta}
	= \beta\mathbb V_{p_{d}}\{h\} 
\end{IEEEeqnarray}
which means
\begin{IEEEeqnarray}{Cl}
	\mathbb V_{p_{d}}\{h\}
	= \frac{1}{\beta}\zeta'(\hat\lambda,\hat\Lambda).
\end{IEEEeqnarray}
In conclusion, the mean and variance of $p_{d}(h;\hat\lambda,\hat\Lambda)$ are given by
\begin{IEEEeqnarray}{Cl}
	\mu_{p_{d}} 
	&= \mathbb E_{p_{d}}\{h\} 
	= \zeta(\hat\lambda,\hat\Lambda)
	\IEEEyesnumber\IEEEyessubnumber*\\
	\Sigma_{p_{d}} 
	&= \mathbb V_{p_{d}}\{h\}
	= \frac{1}{\beta} \zeta'(\hat\lambda,\hat\Lambda) .
\end{IEEEeqnarray}
Define a Gaussian PDF $p(h;\theta,\Theta)$ which has the same mean and variance as $p_{d}(h;\hat\theta,\hat\Theta)$, \textit{i.e.}, $\mu_p = \mu_{p_{d}}$, $\Sigma_p = \Sigma_{p_{d}}$. 
Further define
\begin{IEEEeqnarray}{Cl}
	\upsilon(\hat\lambda,\hat\Lambda) 
	&= \big(\zeta'(\hat\lambda,\hat\Lambda)\big)^{-1} \zeta(\hat\lambda,\hat\Lambda) \notag\\ 
	&=
	\begin{cases}
		(|\hat\lambda|-\kappa)  \text{sign}(\hat\lambda), \quad |\hat\lambda|>\kappa\\
		0, \quad \, \text{otherwise}
	\end{cases}
	\IEEEyesnumber\IEEEyessubnumber*\\
	\Upsilon(\hat\lambda,\hat\Lambda) 
	&= \big(\zeta'(\hat\lambda,\hat\Lambda)\big)^{-1} \notag\\
	&=
	\begin{cases}
		\hat\Lambda, \quad |\hat\lambda|>\kappa\\
		\infty, \quad \, \text{otherwise}
	\end{cases}
\end{IEEEeqnarray}
Then we have that the natural parameters of $p\big(h;\theta,\Theta\big)$ are
$\theta = \beta\lambda$, 
$\Theta = -\frac{\beta}{2}\Lambda$, 
where 
$\lambda = \upsilon(\hat\lambda,\hat\Lambda)$, 
$\Lambda = \Upsilon(\hat\lambda,\hat\Lambda)$.


\section{Proof of Theorem \ref{th: theorem beliefs IGA}}
\label{appendice: theorem beliefs IGA}
The natural parameters of the auxiliary point $p(\mathbf h;\bm\theta_m,\bm\Theta_m)$ are
\begin{IEEEeqnarray}{Cl}
	\bm\theta_m = \beta\left(\bm\lambda_m+\mathbf a_m y_m\right)
	\IEEEyesnumber\IEEEyessubnumber*\\
	\bm\Theta_m = - \frac{\beta}{2}\left(\bm\Lambda_m + \mathbf a_m\mathbf a_m^T\right) .
\end{IEEEeqnarray}
Transforming the natural parameters of auxiliary points into expectation parameters, we get
\begin{IEEEeqnarray}{Cl}
	\bm\mu_m
	&= - \frac{1}{2} \bm\Theta_m^{-1} \bm\theta_m\notag\\
	&= \left(\bm\Lambda_m + \mathbf a_m\mathbf a_m^T\right)^{-1}
	\left(\bm\lambda_m + \mathbf a_m y_m\right)
	\IEEEyesnumber\IEEEyessubnumber*\\
	\bm\Sigma_m
	&= - \frac{1}{2} \bm\Theta_m^{-1}
	= \frac{1}{\beta}\left(\bm\Lambda_m + \mathbf a_m\mathbf a_m^T\right)^{-1} .
\end{IEEEeqnarray}
By $m$-projecting the auxiliary point to the target manifold, we obtain the mean and covariance of the $m$-projection as
\begin{IEEEeqnarray}{ClCl}
	\bm\mu_m^0&=\bm\mu_m,\quad
	&\bm\Sigma_m^0&=\mathbf I\odot\bm\Sigma_m .
\end{IEEEeqnarray}
The natural parameters of the $m$-projection is further obtained as
\begin{IEEEeqnarray}{Cl}
	\subnumberinglabel{SpacetThetam0}
	\bm\theta_m^0
	&= \left(\bm\Sigma_m^0\right)^{-1}\bm\mu_m^0\notag\\
	&= \left(\mathbf I \odot \bm\Theta_m^{-1}\right)^{-1} \bm\Theta_m^{-1} \bm\theta_m\\
	\bm\Theta_m^0
	&= - \frac{1}{2}\left(\bm\Sigma_m^0\right)^{-1} \notag\\
	&= \left(\mathbf I\odot\bm\Theta_m^{-1}\right)^{-1} .
\end{IEEEeqnarray}
Using the Sherman-Morrison formula \cite{sherman1950adjustment} to calculate the matrix inversion in \eqref{SpacetThetam0} yields
\begin{IEEEeqnarray}{Cl}
	\bm\Theta_m^{-1} 
	&= - \frac{2}{\beta} \left(\bm\Lambda_m + \mathbf a_m\mathbf a_m^T\right)^{-1}\notag\\
	&= - \frac{2}{\beta} \left(\bm\Lambda_m^{-1}-\frac{\bm\Lambda_m^{-1}\mathbf a_m\mathbf a_m^T\bm\Lambda_m^{-1}}{1 + \mathbf a_m^T\bm\Lambda_m^{-1}\mathbf a_m}\right)\notag\\
	&= - \frac{2}{\beta} \left(\bm\Lambda_m^{-1}-r_m^{-1}\bm\Lambda_m^{-1}\mathbf a_m\mathbf a_m^T\bm\Lambda_m^{-1}\right)
\end{IEEEeqnarray}
where $r_m=1+\mathbf a_m^T\bm\Lambda_m^{-1}\mathbf a_m$. 
The $m$-projection $\bm\Theta_m^0$ in  \eqref{SpacetThetam0} is then obtained as
\begin{IEEEeqnarray}{Cl}
	\bm\Theta_m^0 &= \left(\mathbf I\odot \bm\Theta_m^{-1}\right)^{-1}\notag\\
	&= - \frac{\beta}{2} \left(\bm\Lambda_m^{-1} 
	- r_m^{-1}\bm\Lambda_m^{-1} 
	\left(\mathbf I\odot\mathbf a_m\mathbf a_m^T\right)
	\bm\Lambda_m^{-1}\right)^{-1}\notag\\
	&= - \frac{\beta}{2} \left(\bm\Lambda_m^{-1} 
	- r_m^{-1}\bm\Lambda_m^{-1}\mathbf L_m\bm\Lambda_m^{-1}\right)^{-1}\notag\\
	&\overset{(c)}{=} - \frac{\beta}{2} \left(\bm\Lambda_m 
	- \left(- r_m\mathbf L_m^{-1} + \bm\Lambda_m^{-1}\bm\Lambda_m\bm\Lambda_m^{-1}\right)^{-1}
	\right)\notag\\
	&= - \frac{\beta}{2} \left(\bm\Lambda_m+\left(r_m\mathbf L_m^{-1}-\bm\Lambda_m^{-1}\right)^{-1}\right)
\end{IEEEeqnarray}
where $\mathbf L_m = \mathbf I\odot\mathbf a_m\mathbf a_m^T$ and $\overset{(c)}{=}$ is from the Woodbury identity \cite{higham2002accuracy}.
From $\bm\Theta_m^0
= -\frac{\beta}{2}\bm\Lambda_m^0 $, we have that the belief $\bm\Xi_m$ is
\begin{IEEEeqnarray}{Cl}
	\label{eq:SpaceXi}
	\bm\Xi_m
	&= \bm\Lambda_m^0
	- \bm\Lambda_m\notag\\ 
	&= \left(r_m\mathbf L_m^{-1}-\bm\Lambda_m^{-1}\right)^{-1} .
\end{IEEEeqnarray}

Next we calculate the $m$-projection $\bm\theta_m^0$ in \eqref{SpacetThetam0}. With $r_m=1+\mathbf a_m^T\bm\Lambda_m^{-1}\mathbf a_m$, it follows that
\begin{IEEEeqnarray}{Cl}
	\bm\theta_m^0
	&=(\mathbf I \odot \bm\Theta_m^{-1})^{-1} \bm\Theta_m^{-1} \bm\theta_m\notag\\
	&=\beta\left(\mathbf I - r_m^{-1}\mathbf L_m\bm\Lambda_m^{-1}\right)^{-1}
	(\mathbf I - r_m^{-1}\mathbf a_m\mathbf a_m^T\bm\Lambda_m^{-1})
	\left(\bm\lambda_m+\mathbf a_m y_m\right)\notag\\
	&=\beta\left(\mathbf I - r_m^{-1}\mathbf L_m\bm\Lambda_m^{-1}\right)^{-1}
	\left((\mathbf I - r_m^{-1}\mathbf a_m\mathbf a_m^T\bm\Lambda_m^{-1})\bm\lambda_m + (\mathbf a_m y_m - r_m^{-1}\mathbf a_m\mathbf a_m^T\bm\Lambda_m^{-1}\mathbf a_m y_m)\right)\notag\\
	&=\beta\left(\mathbf I - r_m^{-1}\mathbf L_m\bm\Lambda_m^{-1}\right)^{-1}
	\left((\mathbf I - r_m^{-1}\mathbf a_m\mathbf a_m^T\bm\Lambda_m^{-1})\bm\lambda_m 
	+ (1 - r_m^{-1}\mathbf a_m^T \bm\Lambda_m^{-1}\mathbf a_m)\mathbf a_m y_m\right)\notag\\
	&=\beta\left(\mathbf I - r_m^{-1}\mathbf L_m\bm\Lambda_m^{-1}\right)^{-1}
	\left((\mathbf I - r_m^{-1}\mathbf a_m\mathbf a_m^T\bm\Lambda_m^{-1})\bm\lambda_m + r_m^{-1}\mathbf a_m y_m\right) .
\end{IEEEeqnarray}
From $\bm\theta_m^0
= \beta\bm\lambda_m^0 $, we have that the belief $\bm\xi_m$ is
\begin{IEEEeqnarray}{Cl}
	\label{eq:Spacexi}
	\bm\xi_m
	&= \bm\lambda_m^0 - \bm\lambda_m\notag\\
	&=\left(\mathbf I - r_m^{-1}\mathbf L_m\bm\Lambda_m^{-1}\right)^{-1}
	\left((\mathbf I - r_m^{-1}\mathbf a_m\mathbf a_m^T\bm\Lambda_m^{-1})\bm\lambda_m + r_m^{-1}\mathbf a_m y_m\right) - \bm\lambda_m\notag\\
	&=\left(\mathbf I - r_m^{-1}\mathbf L_m\bm\Lambda_m^{-1}\right)^{-1}
	\left((\mathbf I - r_m^{-1}\mathbf a_m\mathbf a_m^T\bm\Lambda_m^{-1})\bm\lambda_m 
	- (\mathbf I - r_m^{-1}\mathbf L_m\bm\Lambda_m^{-1})\bm\lambda_m + r_m^{-1}\mathbf a_m y_m\right)\notag\\
	&= \left(r_m\mathbf I - \mathbf L_m\bm\Lambda_m^{-1} \right)^{-1}
	\left((\mathbf L_m-\mathbf a_m\mathbf a_m^T)\bm\Lambda_m^{-1}\bm\lambda_m + \mathbf a_m y_m\right) .
\end{IEEEeqnarray}

\section{Proof of Theorem \ref{th: theorem AMP eq AIGA}}
\label{appendice: theorem AMP eq AIGA}
We use mathematical induction to prove the theorem.
Suppose that at the $t$-th iteration,
\eqref{eq:AMPtauzup} and \eqref{eq:IGLambda0up} are equivalent, 
\eqref{eq:AMPztup} and \eqref{eq:IGztup} are equivalent, 
\eqref{eq:AMPmutup} and \eqref{eq:IGlamhatup}, \eqref{eq:IGLamhatup} are equivalent. 
Under the above assumptions, we have that $\bm\mu^t = (\bm\Lambda_0^t)^{-1} \bm\lambda_0^t$, and $\mathbf z^{t-1}$ in the AMP algorithm and $\mathbf z^{t-1}$ in the AIGA are equal, \textit{i.e.}, $\mathbf z_{\text{AMP}}^{t-1} = \mathbf z_{\text{AIGA}}^{t-1}$. 

According to \cite{Donoho2010}, $\tau_z^t$ in the AMP algorithm is an approximation to $T_{m \rightarrow n}^t$ , while $T_{m \rightarrow n}^t$ is defined in \eqref{eq:Tmton}.
From \eqref{eq:Xiinveqtaumn} we have $[\bm\Xi_m^t]_{nn}^{-1}$ and $(1 + T_{m \rightarrow n}^t)/a_{mn}^2$ correspond to the messages from factor node $y_m$ to variable node $h_n$ in the IGA and the MP algorithm, respectively. 
From $\hat\Lambda_0^t = M\Xi_0^{t-1}$, $\Xi_0^{t-1}$ being an approximation to $[\bm\Xi_m^{t-1}]_{nn}$, $\tau_z^{t-1}$ being an approximation to $T_{m \rightarrow n}^{t-1}$, and $a_{mn}^2
= O(1/M)$, we have 
\begin{IEEEeqnarray}{Cl}
	\frac{1}{M} [\bm\Xi_m^{t-1}]_{nn}^{-1}
	\rightarrow (M \Xi_0^{t-1})^{-1} 
	= (\hat\Lambda_0^t)^{-1}
\end{IEEEeqnarray}
and
\begin{IEEEeqnarray}{Cl}
	\frac{1}{M} \big(1 + T_{m \rightarrow n}^{t-1}\big)
	 \frac{1}{a_{mn}^2}
	\rightarrow 1 + \tau_z^{t-1}.
\end{IEEEeqnarray}
Thus $(\hat\Lambda_0^t)^{-1}$ and $1+\tau_z^{t-1}$ are both approximations to the same message as described in Section \ref{subsec:MP_IGA}. 
Assuming \eqref{eq:AMPtauzup} is equivalent to \eqref{eq:IGLambda0up} is the same as assuming $(\hat\Lambda_0^t)^{-1} = 1+\tau_z^{t-1}$ at the $t$-th iteration.

First, we present the equivalence of \eqref{eq:AMPtauzup} and \eqref{eq:IGLambda0up} at the $(t+1)$-th iteration. 
Rewrite \eqref{eq:IGLambda0up} as
\begin{IEEEeqnarray}{Cl}
	(\hat\Lambda_0^{t+1})^{-1}
	&= 1 + \frac{1}{M} 
	{\rm tr} \big((\bm\Lambda_0^t)^{-1} \big)
\end{IEEEeqnarray}
Note that $\bm\Lambda_0^t
= \Upsilon (\hat{\bm\lambda}_0^t, \hat{\Lambda}_0^t)$, 
$\hat{\bm\lambda}_0^t =\hat{\Lambda}_0^t \big( \mathbf A^T \mathbf z^{t-1}
+ (\bm\Lambda_0^{t-1})^{-1} {\bm\lambda}_0^{t-1} \big)$
and $\delta = M/N$.
Denote the ratio of elements in $\hat{\bm\lambda}_0^t$ whose absolute value is larger than $\kappa$ as $\iota_1^t$, then $\frac{1}{M}{ \rm tr} (\bm\Lambda_0^t)^{-1} = \frac{\iota_1^t}{\delta} (\hat\Lambda_0^t)^{-1}$, and
\begin{IEEEeqnarray}{Cl}
	(\hat\Lambda_0^{t+1})^{-1}
	= 1 + \frac{\iota_1^t}{\delta} 
	(\hat\Lambda_0^t)^{-1}
\end{IEEEeqnarray}
Rewrite \eqref{eq:AMPtauzup} as
\begin{IEEEeqnarray}{Cl}
	&\quad 1 + \tau_z^t  = 1 + \frac{1 + \tau_z^{t-1}}{\delta} 
	\left\langle \eta'\left(\mathbf A^T \mathbf z^{t-1} + \bm\mu^{t-1},\kappa(1 + \tau_z^{t-1})\right) \right\rangle
	\IEEEeqnarraynumspace
\end{IEEEeqnarray}
Similarly, denote the ratio of elements in $\mathbf A^T \mathbf z^{t-1} + \bm\mu^{t-1}$ whose absolute value is larger than $\kappa(1 + \tau_z^{t-1})$ as $\iota_2^t$, then
\begin{IEEEeqnarray}{Cl}
	1 + \tau_z^t 
	&= 1 + \frac{\iota_2^t}{\delta} (1 + \tau_z^{t-1})
\end{IEEEeqnarray}
Next, we verify whether $\iota_1$ and $\iota_2$ are consistent.
$\iota_1$ and $\iota_2$ are determined by the thresholds and arguments of the functions $\eta'$ and $\Upsilon$, respectively.
Suppose that $| [ \hat{\bm\lambda}_0^t ]_n | > \kappa$ for $n$. 
From $\hat\Lambda_0^t > 0$, we obtain that 
$\left| [ \mathbf A^T \mathbf z^{t-1}
+ (\bm\Lambda_0^{t-1})^{-1} \bm\lambda_0^{t-1} ]_n \right| > \kappa  (\hat\Lambda_0^t)^{-1}$, 
Then when $(\hat\Lambda_0^t)^{-1} = 1+\tau_z^{t-1}$, we have $\iota_1^t = \iota_2^t \triangleq \iota^t$, and thus $(\hat\Lambda_0^{t+1})^{-1} = 1+\tau_z^t$. 
Therefore, with the assumptions holding, \eqref{eq:AMPtauzup} and \eqref{eq:IGLambda0up} are equivalent at the $(t+1)$-th iteration.

Next, we present the equivalence of \eqref{eq:AMPztup} and \eqref{eq:IGztup} under the above assumptions at the $(t+1)$-th iteration.
Based on the above derivation, we can get $\iota_1^t = \iota_2^t \triangleq \iota^t$.
Therefore, with the assumptions holding, $\mathbf z_{\text{AMP}}^t = \mathbf z_{\text{AIGA}}^t$, so \eqref{eq:AMPztup} and \eqref{eq:IGztup} are equivalent at the $(t+1)$-th iteration.

Then, we present the equivalence of \eqref{eq:AMPmutup} and \eqref{eq:IGlamhatup}, \eqref{eq:IGLamhatup} at the $(t+1)$-th iteration. 
Based on the above derivation, the functions $\eta$, $\eta'$ have the same thresholds as $\upsilon$, $\Upsilon$ and $\zeta$, $\zeta'$. From $\upsilon(\hat\lambda,\hat\Lambda) = \big(\zeta'(\hat\lambda,\hat\Lambda)\big)^{-1} \zeta(\hat\lambda,\hat\Lambda)$ 
and 
$\Upsilon(\hat\lambda,\hat\Lambda) =  \big(\zeta'(\hat\lambda,\hat\Lambda)\big)^{-1}$, 
we get
$\zeta(\hat\lambda,\hat\Lambda) =  \big(\Upsilon(\hat\lambda,\hat\Lambda)\big)^{-1} \upsilon(\hat\lambda,\hat\Lambda)$, and thus 
\begin{IEEEeqnarray}{Cl}
	&\quad (\bm\Lambda_0^{t+1})^{-1} \bm\lambda_0^{t+1}  = \zeta \left(\hat\Lambda_0^{t+1} \mathbf A^T \mathbf z^t
	+ \hat\Lambda_0^{t+1}     
	(\bm\Lambda_0^t)^{-1} \bm\lambda_0^t, \hat\Lambda_0^{t+1} \right)
\end{IEEEeqnarray}
From the assumption we know that $(\bm\Lambda_0^t)^{-1} \bm\lambda_0^t = \bm\mu^t$. 
Suppose the $n$-th element of  $(\bm\Lambda_0^{t+1})^{-1} \bm\lambda_0^{t+1}$ is not zero, and denote $\mathbf a_n = [\mathbf A]_{:,n}$, 
then from $\hat\Lambda_0^{t+1} > 0$ and the definition of function $\zeta$ in AIGA, we have
\begin{IEEEeqnarray}{Cl}
	 \quad \left[ (\bm\Lambda_0^{t+1})^{-1} \bm\lambda_0^{t+1} \right]_n  
	&= \zeta \left(\hat\Lambda_0^{t+1} \mathbf a_n^T \mathbf z^t
	+ \hat\Lambda_0^{t+1}     
	\left[(\bm\Lambda_0^t)^{-1} \bm\lambda_0^t\right]_n, \hat\Lambda_0^{t+1} \right) \notag\\
	&= \big(\hat\Lambda_0^{t+1}\big)^{-1} \left( \left| \hat\Lambda_0^{t+1} \mathbf a_n^T \mathbf z^t
	+ \hat\Lambda_0^{t+1}     
	[\bm\mu^t]_n \right| - \kappa \right)   \text{sign}(\hat\Lambda_0^{t+1} \mathbf a_n^T \mathbf z^t
	+ \hat\Lambda_0^{t+1}     
	[\bm\mu^t]_n) \notag\\
	&= \left( \left| \mathbf a_n^T \mathbf z^t
	+ [\bm\mu^t]_n \right| - \kappa (\hat\Lambda_0^{t+1})^{-1} \right)  \text{sign}(\mathbf a_n^T \mathbf z^t
	+ [\bm\mu^t]_n)
\end{IEEEeqnarray}
and in the AMP algorithm we have 
\begin{IEEEeqnarray}{Cl}
	[ \bm\mu^{t+1}]_n 
	&= \eta \left(\mathbf a_n^T \mathbf z^t
	+ [\bm\mu^t]_n, \kappa(1 + \tau_z^t) \right) \notag\\
	&= \left( \left| \mathbf a_n^T \mathbf z^t
	+ [\bm\mu^t]_n \right| - \kappa(1 + \tau_z^t) \right)  \text{sign}(\mathbf a_n^T \mathbf z^t
	+ [\bm\mu^t]_n).
\end{IEEEeqnarray}
From $(\hat\Lambda_0^{t+1})^{-1} = 1+\tau_z^t$, we know $\left[ (\bm\Lambda_0^{t+1})^{-1} \bm\lambda_0^{t+1} \right]_n = [ \bm\mu^{t+1}]_n$. 
Therefore, with the assumptions holding, \eqref{eq:AMPmutup} and \eqref{eq:IGlamhatup}, \eqref{eq:IGLamhatup} are equivalent at the $(t+1)$-th iteration.

Finally, when $t=0$, in the AMP algorithm we initialize $\bm\mu^t = \mathbf 0$, $\mathbf z_{\text{AMP}}^{t-1} = \mathbf y$, and chose appropriate $\gamma^{t-1}$, while in the AIGA we initialize $\bm\lambda_0^t = \mathbf 0$, $\bm\Lambda_0^t = \infty$, $\mathbf z_{\text{AIGA}}^{t-1} = \mathbf y$, and chose appropriate $\hat\Lambda_0^t$. 
The initial values of the two algorithms satisfy the relationships $\bm\mu^t = (\bm\Lambda_0^t)^{-1} \bm\lambda_0^t$ and $\mathbf z_{\text{AMP}}^{t-1} = \mathbf z_{\text{AIGA}}^{t-1}$. 
Also notice that $\gamma^{t-1} = \kappa\tau_z^{t-1}$, thus we make $(\hat\Lambda_0^t)^{-1} = 1+\tau_z^{t-1}$ at $t=0$.
For example, let $\hat\Lambda_0^t = 1/2$ and $\tau_z^{t-1} = 1$.
Thus, the initial conditions of the AMP algorithm and the AIGA are the same, \textit{i.e.}, they are equivalent at $t=0$. 
In conclusion, the AMP algorithm and the AIGA are equivalent.

\bibliographystyle{IEEEtran}
\bibliography{IEEEabrv,this_reference}

\end{document}